\documentclass[]{interact}
\usepackage{epstopdf}
\usepackage[caption=false]{subfig}

\usepackage[longnamesfirst,sort]{natbib}
\bibpunct[, ]{(}{)}{;}{a}{,}{,}
\usepackage{natbib}

\usepackage{algorithmic}
\usepackage{algorithm}
\usepackage{amsmath}
\usepackage{amssymb}
\usepackage{caption}
\usepackage{mathtools}
\usepackage{changes}
\usepackage{multirow}
\usepackage{verbatim}
\usepackage{mathrsfs}
\usepackage{graphicx}
\usepackage{amsmath,amsthm,bm,mathrsfs}

\theoremstyle{plain}
\newtheorem{theorem}{Theorem}[section]
\newtheorem{lemma}[theorem]{Lemma}

\theoremstyle{definition}
\newtheorem{definition}[theorem]{Definition}

\theoremstyle{remark}
\newtheorem{remark}{Remark}

\begin{document}

\title{Time- and frequency-limited $\mathcal{H}_2$-optimal model order reduction of bilinear control systems}

\author{
\name{Umair~Zulfiqar\textsuperscript{a}\thanks{CONTACT Umair~Zulfiqar. Email: umair.zulfiqar@research.uwa.edu.au}, Victor Sreeram\textsuperscript{a}, Mian~Ilyas~Ahmad\textsuperscript{b}, and Xin~Du\textsuperscript{c}}
\affil{\textsuperscript{a}School of Electrical, Electronics and Computer Engineering, The University of Western Australia (UWA), Perth, Australia; \textsuperscript{b}Research Centre for Modelling and Simulation, National University of Sciences and Technology (NUST), Islamabad, Pakistan; \textsuperscript{c}School of Mechatronic Engineering and Automation, and Shanghai Key Laboratory of Power Station Automation Technology, Shanghai University, Shanghai, China}
}

\maketitle

\begin{abstract}
In the time- and frequency-limited model order reduction, a reduced-order approximation of the original high-order model is sought to ensure superior accuracy in some desired time and frequency intervals. We first consider the time-limited $\mathcal{H}_2$-optimal model order reduction problem for bilinear control systems and derive first-order optimality conditions that a local optimum reduced-order model should satisfy. We then propose a heuristic algorithm that generates a reduced-order model, which tends to achieve these optimality conditions. The frequency-limited and the time-limited $\mathcal{H}_2$-pseudo-optimal model reduction problems are also considered wherein we restrict our focus on constructing a reduced-order model that satisfies a subset of the respective optimality conditions for the local optimum. Two new algorithms have been proposed that enforce two out of four optimality conditions on the reduced-order model upon convergence. The algorithms are tested on three numerical examples to validate the theoretical results presented in the paper. The numerical results confirm the efficacy of the proposed algorithms.
\end{abstract}


\section{Introduction}
The dynamic behaviour of a physical system is often studied by developing a mathematical model that effectively encompasses its physical characteristics. Conventionally, a direct numerical simulation is conducted using the mathematical model of the system to study its behavior. However, in many applications, we deal with large scale models for which direct simulations are computationally expensive. This motivates the use of model order reduction (MOR) algorithms that generate a reduced-order approximation of the original large-scale model, which effectively mimics its characteristics at a significantly low computational cost \citep{schilders2008model}. In this paper, we discuss the problem of MOR for bilinear systems, which is a special class of nonlinear systems. The strength of bilinear systems lies in the fact that much of its system theory can be developed from linear systems \citep{isidori1973realization}. Bilinear systems find applications in a variety of practical problems like electrical networks, heat transfer, hydraulic systems, fluid flow, and chemical process \citep{mohler1973bilinear,rugh1981nonlinear}. They also find applications in stochastic control problems \citep{benner2011lyapunov,hartmann2013balanced}.

Balanced truncation (BT) is a well-used classical MOR technique for linear systems \citep{moore1981principal}. The approach is famous for its good approximation accuracy, stability preservation, and \textit{apriori} error bound expression. Its theory is extended to bilinear systems in \citep{hsu1983realization,al1993new,zhang2003gramians}. A new definition of the system gramians for bilinear systems is presented in \citep{benner2017truncated}, and its connection with the energy functionals is established. Then the BT method is performed using this new definition of the system gramians to obtain the reduced-order model (ROM).

In the frequency-limited and time-limited MOR scenarios, the goal is to obtain a superior approximation accuracy within some desired frequency and time intervals instead of trying to maintain accuracy over the entire frequency and time ranges. For linear systems, the BT method is generalized to the frequency-limited BT (FLBT) and the time-limited BT (TLBT) in \citep{gawronski1990model}. Several extensions of the FLBT and the TLBT have been reported in the literature to reduce the computational cost \citep{jazlan2015cross,benner2016frequency,kurschner2018balanced}. The FLBT and the TLBT have been extended to bilinear systems in \citep{shaker2013frequency} and \citep{shaker2014time}, respectively.

The $\mathcal{H}_2$-optimal MOR problem has received a lot of attention in the literature wherein the ROM satisfies an optimization criteria in $\mathcal{H}_2$-norm. In \citep{wilson1970optimum}, first-order optimality conditions for the $\mathcal{H}_2$-optimal MOR of a linear system are derived that serve as a foundation for most of the algorithms presented for this problem. The interpolation-based framework for obtaining a ROM that satisfies these optimality conditions is proposed in \citep{gugercin2008h_2,van2008h2}. The $\mathcal{H}_2$-optimal MOR for bilinear systems is considered in \citep{zhang2002h2}. For more general multi-input multi-output (MIMO) bilinear systems, the problem is addressed using tangential interpolation theory as discussed in \citep{breiten2012interpolation,flagg2015multipoint}.

In the $\mathcal{H}_2$-pseudo-optimal MOR problem, a subset of first-order optimality conditions is satisfied instead of the full set to guarantee some other properties like stability. This is the central theme of the work presented in \citep{wolf2013h,wolf2014h,panzer2014model} for linear systems. In \citep{cruz2016interpolation}, the $\mathcal{H}_2$-pseudo-optimal MOR theory is extended to bilinear systems.

The frequency-limited scenario of the $\mathcal{H}_2$-optimal MOR problem for linear systems is considered in \citep{petersson2014model}, and first-order optimality conditions for the local optimum are derived. The optimal ROM in these algorithms is obtained using nonlinear optimization algorithms. In \citep{vuillemin2014frequency}, these optimality conditions are described as bi-tangential Hermite interpolation conditions, and a descent-based algorithm is presented. In \citep{vuillemin2013h2}, the algorithms in \citep{xu2011optimal} and \citep{gugercin2008iterative} are heuristically generalized to the frequency-limited MOR scenario. This algorithm does not satisfy any optimality condition but provides good approximation accuracy. In \citep{zulfiqar2020frequency}, the frequency-limited case of the $\mathcal{H}_2$-pseudo-optimal MOR is considered, and iteration-free algorithms are presented that generate a ROM, which satisfies a subset of first-order optimality conditions. The frequency-limited $\mathcal{H}_2$-optimal MOR problem for the bilinear systems is considered in \citep{xu2017approach}, and first-order optimality conditions are derived. Then a heuristic algorithm (similar to the one in \citep{vuillemin2013h2}) is presented that tends to achieve these optimality conditions.

The time-limited scenario of the $\mathcal{H}_2$-optimal MOR problem for linear systems is considered in \citep{goyal2019time}, and first-order optimality conditions for the local optimum are derived. A heuristic generalization of \citep{xu2011optimal} is also presented in \citep{goyal2019time} that tends to achieve these optimality conditions. In \citep{sinani2019h2}, the optimality conditions are described as bi-tangential Hermite interpolation conditions, and a descent-based algorithm is presented. The time-limited case of the $\mathcal{H}_2$-pseudo-optimal MOR problem is considered in \citep{zulfiqar2019time}, and iteration-free algorithms are presented that generates a ROM, which satisfies a subset of first-order optimality conditions. To the best of our knowledge, the time-limited $\mathcal{H}_2$-optimal MOR problem for bilinear systems is not considered so far in the existing literature.

In this paper, we first define the time-limited $\mathcal{H}_2$-norm for bilinear systems. We then formulate the time-limited $\mathcal{H}_2$-optimal MOR problem using this definition. Then we derive first-order optimality conditions for this problem and give a heuristic algorithm (similar to the one in \citep{goyal2019time}) that construct a ROM, which tends to achieve these optimality conditions. Also, we discuss the reasons why our algorithm and the algorithm in \citep{xu2017approach} may not achieve first-order optimality conditions for their respective problems. In addition, the pseudo-optimal cases are considered for both the time-limited and frequency-limited $\mathcal{H}_2$-MOR problems. Two new algorithms are proposed that achieve a subset of first-order optimality conditions. Lastly, we validate the theory developed in this paper with the help of three numerical examples.

\section{Preliminaries}

Consider a bilinear control system $\Sigma$ with the following state-space equations
\begin{align}
\Sigma :
\begin{dcases}
\dot{x}(t)&=Ax(t)+\sum_{k=1}^{m}N_kx(t)u_k(t)+Bu(t),\\
y(t)&=Cx(t)
\end{dcases}\label{eq:1}
\end{align} where $A\in \mathbb{R}^{n\times n}$, $N_k\in \mathbb{R}^{n\times n}$, $B\in\mathbb{R}^{n\times m}$, and $C\in\mathbb{R}^{p\times n}$. Moreover, $x(t)$, $u(t)$, $y(t)$ are states, control inputs, and outputs, respectively. If the initial conditions are zero, i.e., $x(0)=0$, the output $y(t)$ can be written as the following Volterra series, i.e.,
\begin{align}
y(t)=\sum_{i=1}^{+\infty}\int_{0}^{t}\int_{t_1}^{t}\cdots\int_{t_{i-1}}^{t}\sum_{k_1,k_2,\cdots,k_i}^{m}h_i^{(k_1,k_2,\cdots,k_i)}(t_1,t_2,\cdots,t_i)u_{k_1}(t-t_i)u_{k_2}(t-t_{i-1})\cdots \nonumber\\
u_{k_i}(t-t_{1})dt_1\cdots dt_i\nonumber
\end{align}where
\begin{align}
h_i^{(k_1,k_2,\cdots,k_i)}(t_1,t_2,\cdots,t_i)=Ce^{At_i}N_{k_1}e^{At_{i-1}}N_{k_2}\cdots N_{k_{i-1}}e^{At_1}b_{k_i}\label{1.2}
\end{align} and $b_{k_i}$ is the $k_i^{th}$ column of $B$ \citep{al1993new}. The $i^{th}$ transfer function of $\Sigma$ can be obtained by taking multivariate Laplace transform of (\ref{1.2}), i.e.,
\begin{align}
H_i^{(k_1,k_2,\cdots,k_i)}(s_1,s_2,\cdots,s_i)=C(s_iI-A)^{-1}N_{k_1}N_{k_2}\cdots N_{k_{i-1}}(s_1I-A)^{-1}b_{k_i}.\nonumber
\end{align}
The controllability gramian $P$ \citep{al1993new} for the system in (\ref{eq:1}) is defined as
\begin{align}
P=\sum_{i=1}^{\infty}\int_{0}^{\infty}\cdots\int_{0}^{\infty}P_iP_i^Tdt_1\cdots dt_i\nonumber
\end{align} where
\begin{align}
P_1=e^{At_1}B,\hspace*{2mm}P_i=e^{At_i}\begin{bmatrix}N_1P_{i-1}&\cdots&N_mP_{i-1}\end{bmatrix},\hspace*{2mm}i=2,3,4,\ldots.\nonumber
\end{align}
The observability gramian $Q$ \citep{al1993new} for the system in (\ref{eq:1}) is defined as
\begin{align}
Q=\sum_{i=1}^{\infty}\int_{0}^{\infty}\cdots \int_{0}^{\infty}Q_i^TQ_idt_1 \cdots dt_i\nonumber
\end{align} where
\begin{align}
Q_1=Ce^{At_1},\hspace*{2mm} Q_i=\begin{bmatrix}N_1^TQ_{i-1}^T&\cdots&N_m^TQ_{i-1}^T\end{bmatrix}^Te^{At_i},\hspace*{2mm}i=2,3,4,\ldots.\nonumber
\end{align}
The gramians $P$ and $Q$ solve the following generalized Lyapunov equations
\begin{align}
AP+PA^T+\sum_{k=1}^{m}N_kPN_k^T+BB^T&=0,\nonumber\\
A^TQ+QA+\sum_{k=1}^{m}N_k^TQN_k+C^TC&=0.\nonumber
\end{align}
A detailed discussion about the existence, uniqueness, and solvability of these generalized Lyapunov equations can be found in \citep{benner2011lyapunov}.

The time-limited controllability gramian $P_\tau$ \citep{shaker2014time} for the system in (\ref{eq:1}) within the time interval $[0,\tau]$ sec is defined as
\begin{align}
P_\tau=\sum_{i=1}^{\infty}\int_{0}^{\tau}\cdots\int_{0}^{\tau}P_iP_i^Tdt_1\cdots dt_i.\nonumber
\end{align}
Similarly, the time-limited observability gramian $Q_\tau$ \citep{shaker2014time} for the system in (\ref{eq:1}) within the time interval $[0,\tau]$ sec is defined as
\begin{align}
Q_\tau=\sum_{i=1}^{\infty}\int_{0}^{\tau}\cdots \int_{0}^{\tau}Q_i^TQ_idt_1 \cdots dt_i.\nonumber
\end{align}
The gramians $P_\tau$ and $Q_\tau$ solve the following generalized Lyapunov equations
\begin{align}
AP_\tau+P_\tau A^T+\sum_{k=1}^{m}(N_kP_\tau N_k^T-e^{A\tau}N_kP_\tau N_k^Te^{A^T\tau})+BB^T-e^{A\tau}BB^Te^{A^T\tau}&=0,\label{eq:16}\\
A^TQ_\tau+Q_\tau A+\sum_{k=1}^{m}(N_k^TQ_\tau N_k-e^{A^T\tau}N_k^TQ_\tau N_ke^{A\tau})+C^TC-e^{A^T\tau}C^TCe^{A\tau}&=0.\label{eq:17}
\end{align}
A detailed discussion about the existence, uniqueness, and solvability of these generalized Lyapunov equations can be found in \citep{shaker2014time}.

The frequency-limited controllability gramian $P_\omega$ \citep{shaker2013frequency} for the system in (\ref{eq:1}) within the frequency interval $[0,\omega]$ rad/sec is defined as
\begin{align}
P_\omega=\sum_{i=1}^{\infty}\frac{1}{(2\pi)^i}\int_{-\omega}^{\omega}\cdots\int_{-\omega}^{\omega}\bar{P}_i\bar{P}_i^Td\nu_1\cdots d\nu_i\nonumber
\end{align}
where
\begin{align}
\bar{P}_1=(j\nu_1I-A)^{-1}B,\hspace*{2mm}\bar{P}_i=(j\nu_iI-A)^{-1}\begin{bmatrix}N_1\bar{P}_{i-1}&\cdots&N_m\bar{P}_{i-1}\end{bmatrix},\hspace*{2mm}i=2,3,4,\ldots.\nonumber
\end{align}
Similarly, the frequency-limited observability gramian $Q_\omega$ \citep{shaker2013frequency} for the system in (\ref{eq:1}) within the frequency interval $[0,\omega]$ rad/sec is defined as
\begin{align}
Q_\omega=\sum_{i=1}^{\infty}\frac{1}{(2\pi)^i}\int_{-\omega}^{\omega}\cdots \int_{-\omega}^{\omega}\bar{Q}_i^T\bar{Q}_id\nu_1 \cdots d\nu_i\nonumber
\end{align} where
\begin{align}
\bar{Q}_1=C(j\nu_1I-A)^{-1},\hspace*{2mm} \bar{Q}_i=\begin{bmatrix}N_1^T\bar{Q}_{i-1}^T&\cdots&N_m^T\bar{Q}_{i-1}^T\end{bmatrix}^T(j\nu_iI-A)^{-1},\hspace*{2mm}i=2,3,4,\ldots.\nonumber
\end{align}
The gramians $P_\omega$ and $Q_\omega$ solve the following generalized Lyapunov equations
\begin{align}
AP_\omega+P_\omega A^T+\sum_{k=1}^{m}(F_\omega[A] N_kP_\omega N_k^T&+N_kP_\omega N_k^TF_\omega[A]^T)\nonumber\\
&+F_\omega[A] BB^T+BB^TF_\omega[A]^T=0,\\
A^TQ_\omega+Q_\omega A+\sum_{k=1}^{m}(F_\omega[A]^TN_k^TQ_\omega N_k&+N_k^TQ_\omega N_kF_\omega[A])\nonumber\\
&+F_\omega[A]^T C^TC+C^TCF_\omega[A]=0
\end{align} where $F_\omega[A]=\frac{1}{2\pi}\int_{-\omega}^{\omega}(j\nu I-A)^{-1}d\nu=Real(\frac{j}{\pi}ln(-A-j\omega I))$ \citep{petersson2014model}. A detailed discussion about the existence, uniqueness, and solvability of these generalized Lyapunov equations can be found in \citep{shaker2013frequency}.

The frequency-limited $\mathcal{H}_2$-norm \citep{xu2017approach}, i.e, $\mathcal{H}_{2,\omega}$-norm, of $\Sigma$ within the frequency interval $[0,\omega]$ rad/sec is defined as the following
\begin{align}
  ||\Sigma||_{\mathcal{H}_{2,\omega}}&=\sqrt{\begin{aligned}trace\Big(\sum_{i=1}^{\infty}\frac{1}{(2\pi)^i}\int_{-\omega}^{\omega}\cdots\int_{-\omega}^{\omega}&\sum_{k_1,\cdots,k_i=1}^{m}H_i^{(k_1,k_2,\cdots,k_i)}(j\nu_1,j\nu_2,\cdots,j\nu_i)\\
  &\times\big(H_i^{(k_1,k_2,\cdots,k_i)}(j\nu_1,j\nu_2,\cdots,j\nu_i)\big)^*d\nu_1\cdots d\nu_i\Big)\end{aligned}}\nonumber\\
  &=\sqrt{trace(CP_\omega C^T)}\nonumber\\
  &=\sqrt{\begin{aligned}trace\Big(\sum_{i=1}^{\infty}\frac{1}{(2\pi)^i}\int_{-\omega}^{\omega}\cdots\int_{-\omega}^{\omega}&\sum_{k_1,\cdots,k_i=1}^{m}\big(H_i^{(k_1,k_2,\cdots,k_i)}(j\nu_1,j\nu_2,\cdots,j\nu_i)\big)^*\\
  &\times H_i^{(k_1,k_2,\cdots,k_i)}(j\nu_1,j\nu_2,\cdots,j\nu_i)d\nu_1\cdots d\nu_i\Big)\end{aligned}}\nonumber\\
  &=\sqrt{trace(B^TQ_\omega B)}\nonumber
\end{align} where $\begin{bmatrix}\cdot\end{bmatrix}^*$ represents the Hermitian.
\subsection{Problem Statement}
The MOR problem under consideration is to obtain a ROM $\tilde{\Sigma}$ of $\Sigma$ such that $\tilde{\Sigma}$ accurately mimics $\Sigma$ when used as a surrogate. Let $\tilde{\Sigma}$ is represented by the following state-space equations
\begin{align}
\tilde{\Sigma} :
\begin{dcases}
\dot{\tilde{x}}(t)&=\tilde{A}\tilde{x}(t)+\sum_{k=1}^{m}\tilde{N}_{k}\tilde{x}(t)u_k(t)+\tilde{B}u(t),\\
\tilde{y}(t)&=\tilde{C}\tilde{x}(t)
\end{dcases}\label{eq:2}
\end{align} where $\tilde{A}\in \mathbb{R}^{r\times r}$, $\tilde{N}_{k}\in \mathbb{R}^{r\times r}$, $\tilde{B}\in\mathbb{R}^{r\times m}$, and $\tilde{C}\in\mathbb{R}^{p\times r}$ such that $r\ll n$. In projection-based MOR, $\tilde{\Sigma}$ is computed as the following
\begin{align}
\tilde{A}&=W^TAV,&\tilde{N}_k&=W^TN_kV,&\tilde{B}&=W^TB,&\tilde{C}&=CV
\end{align} where $W^TV=I$ and $V,W\in \mathbb{R}^{n \times r}$ such that their column spans form basis to some specific $r$-dimensional subspaces. The quality of the approximation of $\Sigma$ is quantified by using various norms for the error expression $\Sigma_e=\Sigma-\tilde{\Sigma}$. The error system $\Sigma_e$ has the following state-space realization
\begin{align}
\Sigma_e :
\begin{dcases}
\dot{x_e}(t)&=A_ex_e(t)+\sum_{k=1}^{m}N_{ek}x_e(t)u_k(t)+B_eu(t),\\
y_e(t)&=C_ex_e(t)
\end{dcases}\nonumber
\end{align} where
\begin{align}
A_e=\begin{bmatrix}A&0\\0&\tilde{A}\end{bmatrix},\hspace*{2mm}N_{ek}=\begin{bmatrix}N_k&0\\0&\tilde{N}_k\end{bmatrix},\hspace*{2mm}B_e=\begin{bmatrix}B\\\tilde{B}\end{bmatrix},\hspace*{2mm}C_e=\begin{bmatrix}C&-\tilde{C}\end{bmatrix}.\label{1.23}
\end{align}

In some practical situations, it is desirable to ensure that $\tilde{\Sigma}$ accurately approximates $\Sigma$ within the desired frequency interval $[0,\omega]$ rad/sec. The $\mathcal{H}_{2,\omega}$-norm is generally used to quantify the approximation error in this scenario \citep{xu2017approach}. The $\mathcal{H}_{2,\omega}$-MOR problem is to find a ROM of order $r$ which ensures that $||\Sigma_e||^2_{\mathcal{H}_{2,\omega}}$ is small, i.e., $\underset{\substack{\tilde{\Sigma}\\\textnormal{order}=r}}{\text{min}}||\Sigma_e||^2_{\mathcal{H}_{2,\omega}}$.
Similarly, it is often desirable that the approximation accuracy is good within the desired time interval $[0,\tau]$ sec. We will formulate the definition of time-limited $\mathcal{H}_2$-norm, i.e., $\mathcal{H}_{2,\tau}$-norm, in the next section to quantify the quality of approximation in this scenario. The $\mathcal{H}_{2,\tau}$-MOR problem is to find an ROM of order $r$ which ensures that $||\Sigma_e||^2_{\mathcal{H}_{2,\tau}}$ is small, i.e., $\underset{\substack{\tilde{\Sigma}\\\textnormal{order}=r}}{\text{min}}||\Sigma_e||^2_{\mathcal{H}_{2,\tau}}$.
\subsection{Frequency-limited $\mathcal{H}_2$-optimal MOR}
Let $\tilde{P}_\omega$ and $\tilde{Q}_\omega$ be the frequency-limited controllability and frequency-limited observability gramians of $\tilde{\Sigma}$, respectively. Then the squared $\mathcal{H}_{2,\omega}$-norm of $\Sigma_e$ can be expressed as
\begin{align}
||\Sigma_e||_{\mathcal{H}_{2,\omega}}^2&=trace(CP_\omega C^T-2C\hat{P}_\omega\tilde{C}^T+\tilde{C}\tilde{P}_\omega\tilde{C}^T)\nonumber\\
&=trace(B^TQ_\omega B+2B^T\hat{Q}_\omega\tilde{B}+\tilde{B}^T\tilde{Q}_\omega\tilde{B})\nonumber
\end{align} where $\hat{P}_\omega$ and $\hat{Q}_\omega$ solve the following generalized Sylvester equations
\begin{align}
A\hat{P}_\omega+\hat{P}_\omega\tilde{A}^T+\sum_{k=1}^{m}\big(F_\omega[A]N_k\hat{P}_\omega\tilde{N}_k^T&+N_k\hat{P}_\omega\tilde{N}_k^TF_\omega[\tilde{A}]^T\big)\nonumber\\
&+F_\omega[A]B\tilde{B}^T+B\tilde{B}^TF_\omega[\tilde{A}]^T=0,\nonumber\\
A^T\hat{Q}_\omega+\hat{Q}_\omega\tilde{A}+\sum_{k=1}^{m}\big(F_\omega[A]^TN_k^T\hat{Q}_\omega\tilde{N}_k&+N_k^T\hat{Q}_\omega\tilde{N}_kF_\omega[\tilde{A}]\big)\nonumber\\
&-F_\omega[A]^TC^T\tilde{C}-C^T\tilde{C}F_\omega[\tilde{A}]=0.\nonumber
\end{align}
Let $R_\omega$ and $S_\omega$ solve the following generalized matrix equations
\begin{align}
A^TR_\omega+R_\omega \tilde{A}+\sum_{k=1}^{m}N_k^TF_\omega[A]^TR_\omega \tilde{N}_{k}+\sum_{k=1}^{m}N_k^TR_\omega F_\omega[\tilde{A}]\tilde{N}_{k}-C^T\tilde{C}&=0,\nonumber\\
\tilde{A}^TS_\omega+S_\omega \tilde{A}+\sum_{k=1}^{m}\tilde{N}_{k}^TS_\omega F_\omega[\tilde{A}]\tilde{N}_{k}+\sum_{k=1}^{m}\tilde{N}_{k}^TF_\omega[\tilde{A}]^TS_\omega\tilde{N}_{k}+\tilde{C}^T\tilde{C}&=0.\nonumber
\end{align}
Also, if we define $S_1$, $S_2$, and $\hat{W}_i$ as
\begin{align}
S_1&=\sum_{k=1}^{m}\tilde{N}_{k}\tilde{P}_\omega \tilde{N}_{k}^TS_\omega+\tilde{B}\tilde{B}^TS_\omega,&
S_2&=\sum_{k=1}^{m}\tilde{N}_{k}\hat{P}_{\omega}^TN_k^TR_\omega+\tilde{B}B^TR_\omega,\nonumber\\
\hat{W}_i&=Real\Big[\frac{j}{\pi}L\big(-\tilde{A}-j\omega I, S_i\big)\Big]&&\nonumber
\end{align} in which $L(\cdot,\cdot)$ represents the Frech\'et derivative of the matrix logarithm \citep{higham2008functions}, the first-order optimality conditions for the $\mathcal{H}_{2,\omega}$-optimal MOR are given by
\begin{align}
&R_\omega^T\hat{P}_\omega+S_\omega\tilde{P}_\omega=\hat{W}_1^T+\hat{W}_2^T,\label{1.45}\\
&\sum_{k=1}^{m}\Big(\big(R_\omega^TF_\omega[A]+F_\omega[\tilde{A}]^TR_\omega^T\big)N_k\hat{P}_\omega+\big(F_\omega[\tilde{A}]^TS_\omega+S_\omega F_\omega[\tilde{A}]\big)\tilde{N}_k\tilde{P}_\omega\Big)=0,\\
&\hat{Q}_\omega^TB+\tilde{Q}_\omega\tilde{B}=0,\label{1.47}\\
&C\hat{P}_\omega-\tilde{C}\tilde{P}_\omega=0.\label{1.48}
\end{align}
In \citep{xu2017approach}, an iterative algorithm is presented that generates a ROM, which approximately satisfies the optimality conditions (\ref{1.45})-(\ref{1.48}). We refer to this algorithm as the frequency-limited $\mathcal{H}_2$-MOR algorithm (FLHMORA). Starting with an initial guess of the ROM $(\bar{A},\bar{N}_k,\bar{B},\bar{C})$, the reduction subspaces are updated in each iteration as $V=orth(V_\omega)$, $W=orth(W_\omega)$, and $W=W(V^TW)^{-1}$ until the algorithm converges where
\begin{align}
AV_\omega+V_\omega\bar{A}^T+\sum_{k=1}^{m}(F_\omega[A]N_kV_\omega\bar{N}_k^T&+N_kV_\omega\bar{N}_k^TF_\omega[\bar{A}]^T)\nonumber\\
&+F_\omega[A]B\bar{B}^T+B\bar{B}^TF_\omega[\bar{A}]^T=0,\nonumber\\ A^TW_\omega+W_\omega\bar{A}+\sum_{k=1}^{m}(F_\omega[A]^TN_k^TW_\omega\bar{N}_k&+N_k^TW_\omega\bar{N}_kF_\omega[\bar{A}])\nonumber\\
&-F_\omega[A]^TC^T\bar{C}-C^T\bar{C}F_\omega[\bar{A}]=0.\nonumber
\end{align}
\subsection{TLBT}
Let the time-limited Hankel singular values $\sigma_i$ of $\Sigma$ be defined as $\sigma_i=\sqrt{\lambda_i(P_\tau Q_\tau)}$ where $\lambda_i(\cdot)$ represents the eigenvalues. Heuristically, $\sigma_i$ is the quantitative measure of a state's contribution to the energy transfer within the desired time interval $[0,\tau]$ sec. In the TLBT \citep{shaker2014time}, the states with the negligible time-limited Hankel singular values are truncated. The reduction subspaces $V$ and $W$ are computed as $\tilde{P}_\tau=\tilde{Q}_\tau\approx W^TP_\tau W=V^TQ_\tau V=diag(\sigma_1,\cdots,\sigma_r)$ where $\sigma_1\geq\cdots\geq\sigma_r$ are the $r$ largest time-limited Hankel singular values of $\Sigma$, and $\tilde{P}_\tau$ and $\tilde{Q}_\tau$ are the time-limited controllability and time-limited observability gramians of $\tilde{\Sigma}$.
\section{Main Work}
In this section, we first formulate the definition of $\mathcal{H}_{2,\tau}$-norm for bilinear control systems. Then we derive first-order optimality conditions for the  $\mathcal{H}_{2,\tau}$-optimal MOR problem, i.e., conditions for the local optimum of $||\Sigma_e||^2_{\mathcal{H}_{2,\tau}}$. To ensure these conditions, an extension of the $\mathcal{H}_2$-optimal MOR algorithm (HOMORA) \citep{breiten2012interpolation} has been proposed for the time-limited MOR case that tends to satisfy the derived optimality conditions. The difficulty in enforcing the optimality conditions associated with $\tilde{A}$ and $\tilde{N}_k$ in the $\mathcal{H}_{2,\omega}$- and $\mathcal{H}_{2,\tau}$-optimal MOR are also discussed. Two new algorithms have been proposed to enforce the optimality conditions associated with $\tilde{B}$ and $\tilde{C}$ on the ROM for a fixed choice of $\tilde{A}$ and $\tilde{N}_k$. We end this section with a discussion on the computational aspects of the proposed algorithms.
\subsection{Time-limited $\mathcal{H}_2$-optimal MOR}
The $\mathcal{H}_2$-norm of $\Sigma$ quantifies the power of the output response $y(t)$ to the unit white noise input and is defined over the entire time horizon. If we are only interested in the power of the output response $y(t)$ within a finite time interval $[0,\tau]$ sec, we need to restrict the output response $y(t)$ within that interval. This results in a new norm, which we refer to as $\mathcal{H}_{2,\tau}$-norm. We now mathematically formulate the definition of $\mathcal{H}_{2,\tau}$-norm.
\begin{definition}
The time-limited $\mathcal{H}_2$-norm, i.e., $\mathcal{H}_{2,\tau}$-norm, of a bilinear system $\Sigma$ with a Hurwitz $A$-matrix within the time interval $[0,\tau]$ sec is defined as
\begin{align}
  ||\Sigma||_{\mathcal{H}_{2,\tau}}&=\sqrt{\begin{aligned}trace\Big(\sum_{i=1}^{\infty}\int_{0}^{\tau}\cdots\int_{0}^{\tau}\sum_{k_1,\cdots,k_i=1}^{m}&h_i^{(k_1,k_2,\cdots,k_i)}(t_1,t_2,\cdots,t_i)\\
  &\times\big(h_i^{(k_1,k_2,\cdots,k_i)}(t_1,t_2,\cdots,t_i)\big)^Tdt_1\cdots dt_i\Big)\end{aligned}}\nonumber\\
  &=\sqrt{\begin{aligned}trace\Big(\sum_{i=1}^{\infty}\int_{0}^{\tau}\cdots\int_{0}^{\tau}\sum_{k_1,\cdots,k_i=1}^{m}&\big(h_i^{(k_1,k_2,\cdots,k_i)}(t_1,t_2,\cdots,t_i)\big)^T\\
  &\times h_i^{(k_1,k_2,\cdots,k_i)}(t_1,t_2,\cdots,t_i)dt_1\cdots dt_i\Big).\end{aligned}}\nonumber
\end{align}
\end{definition}
Note that when $N_k=0$ for $k=1,2,\cdots,m$, this $\mathcal{H}_{2,\tau}$-norm reduces to the $\mathcal{H}_{2,\tau}$-norm for linear systems \citep{goyal2019time}. For linear systems, the $\mathcal{H}_{2,\tau}$-norm is related to the time-limited controllability and time-limited observability gramians \citep{goyal2019time}. In the following, we show a similar relation in the bilinear case.
\begin{theorem}
  If $P_\tau$ or $Q_\tau$ exists, the $\mathcal{H}_{2,\tau}$-norm of $\Sigma$ can be expressed in terms of $P_\tau$ or $Q_\tau$ by using
  \begin{align}
  ||\Sigma||_{\mathcal{H}_{2,\tau}}=\sqrt{trace(CP_\tau C^T)}=\sqrt{trace(B^T Q_\tau B)}.\nonumber
  \end{align}
\end{theorem}
\begin{proof}
Considering the duality of $P_\tau$ and $Q_\tau$, we restrict ourselves in proving that $||\Sigma||_{\mathcal{H}_{2,\tau}}=\sqrt{trace(CP_\tau C^T)}$. Let $H_i$ and $P_{i,\tau}$ be defined as
\begin{align}
H_i&=\int_{0}^{\tau}\cdots\int_{0}^{\tau}\sum_{k_1,\cdots,k_i=1}^{m}h_i^{(k_1,k_2,\cdots,k_i)}(t_1,t_2,\cdots,t_i)\times\nonumber\\
&\hspace*{4cm}\big(h_i^{(k_1,k_2,\cdots,k_i)}(t_1,t_2,\cdots,t_i)\big)^Tdt_1\cdots dt_i,\nonumber\\
P_{i,\tau}&=\int_{0}^{\tau}\cdots\int_{0}^{\tau}P_iP_i^Tdt_1\cdots dt_i.\nonumber
\end{align}
Then, $||\Sigma||_{\mathcal{H}_{2,\tau}}=\sqrt{trace(\sum_{i=1}^{\infty}H_i)}$ and $P_\tau=\sum_{i=1}^{\infty}P_{i,\tau}$. Now for $i=1$,
\begin{align}
H_1&=\int_{0}^{\tau}\sum_{k_1=1}^{m}Ce^{At_1}b_{k_1}b_{k_1}^Te^{A^Tt_1}C^Tdt_1.\nonumber
\end{align} Since $\sum_{k_i=1}^{m}b_{k_i}b_{k_i}^T=BB^T$,
\begin{align}
H_1&=\int_{0}^{\tau}Ce^{At_1}BB^Te^{A^Tt_1}C^Tdt_1=CP_{1,\tau}C^T.\nonumber
\end{align}
Similarly, for $i=2$, we have
\begin{align}
H_2&=\int_{0}^{\tau}\int_{0}^{\tau}\sum_{k_1,k_2=1}^{m}Ce^{At_2}N_{k_1}e^{At_1}b_{k_2}b_{k_2}^Te^{A^Tt_1}N_{k_1}^Te^{At_2}C^Tdt_1dt_2\nonumber\\
&=\int_{0}^{\tau}\int_{0}^{\tau}\sum_{k_1=1}^{m}Ce^{At_2}N_{k_1}P_1P_1^TN_{k_1}^Te^{A^Tt_2}C^Tdt_1dt_2\nonumber\\
&=\int_{0}^{\tau}\int_{0}^{\tau}CP_2P_2^TC^Tdt_1dt_2=CP_{2,\tau}C^T.\nonumber
\end{align}
On similar lines, $H_i$ can be defined for $i=3,4,\cdots$ as the following
\begin{align}
H_i&=\int_{0}^{\tau}\cdots\int_{0}^{\tau}\sum_{k_1,\cdots,k_i=1}^{m}Ce^{At_i}\cdots e^{At_1}b_{k_i}b_{k_i}^Te^{A^Tt_1}\cdots e^{A^Tt_i}C^Tdt_1\cdots dt_i\nonumber\\
&=\int_{0}^{\tau}\cdots\int_{0}^{\tau}\sum_{k_1,\cdots,k_{i-1}=1}^{m}Ce^{At_i}N_{k_1}N_{k_{i-2}}e^{At_2}N_{k_{i-1}}P_1P_1^TN_{k_{i-1}}^Te^{A^Tt_2}\cdots \nonumber\\
&\hspace*{8cm}\cdots N_{k_1}^Te^{A^Tt_i}C^Tdt_1\cdots dt_i\nonumber\\
&=\int_{0}^{\tau}\cdots\int_{0}^{\tau}\sum_{k_1,\cdots,k_{i-2}=1}^{m}Ce^{At_i}N_{k_1}\cdots N_{k_{i-3}}e^{At_3}N_{k_{i-2}}P_2P_2^TN_{k_{i-2}}^Te^{A^Tt_3}\cdots \nonumber\\
&\hspace*{8cm}\cdots N_{k_1}^Te^{A^Tt_i}C^Tdt_1\cdots dt_i\nonumber\\
&=\int_{0}^{\tau}\cdots\int_{0}^{\tau}\sum_{k_1=1}^{m}Ce^{At_i}N_{k_1}P_{i-1}P_{i-1}^TN_{k_1}^Te^{A^Tt_i}C^Tdt_1\cdots dt_i\nonumber
\end{align}
\begin{align}
&=\int_{0}^{\tau}\cdots\int_{0}^{\tau}CP_{i}P_{i}^TC^Tdt_1\cdots dt_i=CP_{i,\tau}C^T.\nonumber
\end{align}
Thus $||\Sigma||_{\mathcal{H}_{2,\tau}}=\sqrt{trace(\sum_{i=1}^{\infty}H_i)}=\sqrt{trace(\sum_{i=1}^{\infty}CP_{i,\tau}C^T)}=\sqrt{trace(CP_{\tau}C^T)}$. This completes the proof.\end{proof}
Then the squared $\mathcal{H}_{2,\tau}$-norm of $\Sigma_e$ is given by
\begin{align}
||\Sigma_e||^2_{\mathcal{H}_{2,\tau}}=trace(CP_{e,\tau}C^T)=trace(B^TQ_{e,\tau}B)\nonumber
\end{align} where $P_{e,\tau}$ and $Q_{e,\tau}$ solve the following generalized Lyapunov equations
\begin{align}
A_eP_{e,\tau}+P_{e,\tau}A_e^T+\sum_{k=1}^{m}\big(N_{ek}P_{e,\tau}N_{ek}^T&-e^{A_e\tau}N_{ek}P_{e,\tau}N_{ek}^Te^{A_e^T\tau}\big)\nonumber\\
&+B_eB_e^T-e^{A_e\tau}B_eB_e^Te^{A_e^T\tau}=0,\nonumber\\
A_e^TQ_{e,\tau}+Q_{e,\tau}A_e+\sum_{k=1}^{m}\big(N_{ek}^TQ_{e,\tau}N_{ek}&-e^{A_e^T\tau}N_{ek}^TQ_{e,\tau}N_{ek}e^{A_e\tau}\big)\nonumber\\
&+C_e^TC_e-e^{A_e^T\tau}C_e^TC_ee^{A_e\tau}=0.\nonumber
\end{align}
The matrices $P_{e,\tau}$, $Q_{e,\tau}$, and $e^{A_e\tau}$ can be partitioned according to (\ref{1.23}) to get
\begin{align}
P_{e,\tau}=\begin{bmatrix}P_\tau&\hat{P}_\tau\\\hat{P}_\tau^T&\tilde{P}_\tau\end{bmatrix},\hspace*{4mm} Q_{e,\tau}=\begin{bmatrix}Q_\tau&\hat{Q}_\tau\\\hat{Q}_\tau^T&\tilde{Q}_\tau\end{bmatrix},\hspace*{4mm} \textnormal{and}\hspace*{4mm}e^{A_e\tau}=\begin{bmatrix}e^{A\tau}&0\\0&e^{\tilde{A}\tau}\end{bmatrix}\nonumber
\end{align} where $\tilde{P}_\tau$, $\hat{P}_\tau$, $\tilde{Q}_\tau$, and $\hat{Q}_\tau$ solve the following generalized matrix equations
\begin{align}
\tilde{A}\tilde{P}_\tau+\tilde{P}_\tau\tilde{A}^T+\sum_{k=1}^{m}\big(\tilde{N}_k\tilde{P}_\tau\tilde{N}_k^T&-e^{\tilde{A}\tau}\tilde{N}_k\tilde{P}_\tau\tilde{N}_k^Te^{\tilde{A}^T\tau}\big)\nonumber\\
&+\tilde{B}\tilde{B}^T-e^{\tilde{A}\tau}\tilde{B}\tilde{B}^Te^{\tilde{A}^T\tau}=0,\label{2.10}\\
A\hat{P}_\tau+\hat{P}_\tau\tilde{A}^T+\sum_{k=1}^{m}\big(N_k\hat{P}_\tau\tilde{N}_k^T&-e^{A\tau}N_k\hat{P}_\tau\tilde{N}_k^Te^{\tilde{A}^T\tau}\big)\nonumber\\
&+B\tilde{B}^T-e^{A\tau}B\tilde{B}^Te^{\tilde{A}^T\tau}=0,\label{2.11}\\
\tilde{A}^T\tilde{Q}_\tau+\tilde{Q}_\tau\tilde{A}+\sum_{k=1}^{m}\big(\tilde{N}_k^T\tilde{Q}_\tau\tilde{N}_k&-e^{\tilde{A}^T\tau}\tilde{N}_k^T\tilde{Q}_\tau\tilde{N}_ke^{\tilde{A}\tau}\big)\nonumber\\
&+\tilde{C}^T\tilde{C}-e^{\tilde{A}^T\tau}\tilde{C}^T\tilde{C}e^{\tilde{A}\tau}=0,\nonumber\\
A^T\hat{Q}_\tau+\hat{Q}_\tau\tilde{A}+\sum_{k=1}^{m}\big(N_k^T\hat{Q}_\tau\tilde{N}_k&-e^{A^T\tau}N_k^T\hat{Q}_\tau\tilde{N}_ke^{\tilde{A}\tau}\big)\nonumber\\
&-C^T\tilde{C}+e^{A^T\tau}C^T\tilde{C}e^{\tilde{A}\tau}=0.\nonumber
\end{align}
Accordingly, the squared $\mathcal{H}_{2,\tau}$-norm of $\Sigma_e$ can be expressed as
\begin{align}
||\Sigma_e||_{\mathcal{H}_{2,\tau}}^2&=trace(CP_\tau C^T-2C\hat{P}_\tau\tilde{C}^T+\tilde{C}\tilde{P}_\tau\tilde{C}^T)\nonumber\\
&=trace(B^TQ_\tau B+2B^T\hat{Q}_\tau\tilde{B}+\tilde{B}^T\tilde{Q}_\tau\tilde{B}).\nonumber
\end{align}
\begin{lemma}\label{lemma1}
Let $L$ and $Z$ solve the following generalized matrix equations
\begin{align}
AL+L\tilde{A}^T+\sum_{k=1}^{m}(N_kL\tilde{N}_k^T-e^{A\tau}N_kL\tilde{N}_k^Te^{\tilde{A}^T\tau})+O_1&=0,\nonumber\\
\tilde{A}^TZ+ZA+\sum_{k=1}^{m}(\tilde{N}_k^TZN_k-\tilde{N}_k^Te^{\tilde{A}^T\tau}Ze^{A\tau}N_k)+O_2&=0.\nonumber
\end{align}
Then $trace(O_1Z)=trace(O_2L)$.
\end{lemma}
\begin{proof}We use three main properties of trace:\\
(a) Transpose: $trace(F_1F_2F_3)=trace(F_3^TF_2^TF_1^T)$.\\
(b) Cyclic permutation: $trace(F_1F_2F_3)=trace(F_3F_1F_2)=trace(F_2F_3F_1)$.\\
(c) Addition: $trace(F_1+F_2+F_3)=trace(F_1)+trace(F_2)+trace(F_3).$

Now \begin{align}
trace(O_1Z)&=-trace\big(ALZ+L\tilde{A}^TZ+\sum_{k=1}^{m}(N_kL\tilde{N}_k^TZ-e^{A\tau}N_kL\tilde{N}_k^Te^{\tilde{A}^T\tau}Z)\big)\nonumber\\
&=-trace\Big(\big(ZA+\tilde{A}^TZ+\sum_{k=1}^{m}(\tilde{N}_k^TZN_k-\tilde{N}_k^Te^{\tilde{A}^T\tau}Ze^{A\tau}N_k)\big)L\Big)\nonumber\\
&=trace(O_2L).\nonumber
\end{align} A special case of this lemma is when $A$ is replaced with $\tilde{A}$ as no specific dimension of $A$ is assumed.
\end{proof}
\begin{theorem}
Let $A$ and $\tilde{A}$ be Hurwitz, and $P_\tau$, $\tilde{P}_\tau$, $Q_\tau$, and $\tilde{Q}_\tau$ exist. Let $R_\tau$ and $S_\tau$ solve the following generalized matrix equations
\begin{align}
A^TR_\tau+R_\tau \tilde{A}+\sum_{k=1}^{m}\big(N_k^TR_\tau \tilde{N}_{k}-N_k^Te^{A^T\tau}R_\tau e^{\tilde{A}\tau}\tilde{N}_{k}\big)-C^T\tilde{C}&=0,\label{3.10}\\
\tilde{A}^TS_\tau+S_\tau \tilde{A}+\sum_{k=1}^{m}\big(\tilde{N}_{k}^TS_\tau\tilde{N}_{k}-\tilde{N}_{k}^Te^{\tilde{A}^T\tau}S_\tau e^{\tilde{A}\tau}\tilde{N}_{k}\big)+\tilde{C}^T\tilde{C}&=0.\label{3.11}
\end{align} Then the partial derivatives of the cost function $J(\tilde{A},\tilde{N}_k,\tilde{B},\tilde{C})=||\Sigma_e||_{\mathcal{H}_{2,\tau}}^2$ with respect to $\tilde{A}$, $\tilde{N}_k$, $\tilde{B}$, and $\tilde{C}$ are given by
\begin{align}
\frac{\partial J}{\partial\tilde{A}}&=2(R_\tau^T\hat{P}_\tau+S_\tau\tilde{P}_\tau-Y_\tau),\nonumber\\
\frac{\partial J}{\partial\tilde{N}_k}&=2\sum_{k=1}^{m}(R_\tau^TN_k\hat{P}_\tau+S_\tau\tilde{N}_k\tilde{P}_\tau-Z_\tau\big),\nonumber\\
\frac{\partial J}{\partial\tilde{B}}&=2(\hat{Q}_\tau^TB+\tilde{Q}_\tau\tilde{B}),\nonumber\\
\frac{\partial J}{\partial\tilde{C}}&=2(-C\hat{P}_\tau+\tilde{C}\tilde{P}_\tau),\nonumber
\end{align}
where
\begin{align}
Y_\tau&=\tau( R_\tau^Te^{A\tau}B\tilde{B}^Te^{\tilde{A}^T\tau}+S_\tau e^{\tilde{A}\tau}\tilde{B}\tilde{B}^Te^{\tilde{A}^T\tau})+\tau\sum_{k=1}^{m}(R_\tau^Te^{A\tau}N_k\hat{P}_\tau\tilde{N}_k^Te^{\tilde{A}^T\tau}\nonumber\\
&\hspace*{9.5cm}+S_\tau e^{\tilde{A}\tau}\tilde{N}_k\tilde{P}_\tau\tilde{N}_k^Te^{\tilde{A}^T\tau}),\nonumber\\
Z_\tau&=e^{\tilde{A}^T\tau}R_\tau^Te^{A\tau}N_k\hat{P}_\tau+e^{\tilde{A}^T\tau}S_\tau e^{\tilde{A}\tau}\tilde{N}_k\tilde{P}_\tau.\nonumber
\end{align}
\end{theorem}
\begin{proof}Let us denote the first-order derivative of $J$, $\tilde{P}_\tau$, and $\hat{P}_\tau$ with respect to $\tilde{A}$ as $\Delta_{J}^{\tilde{A}}$, $\Delta_{\tilde{P}_\tau}^{\tilde{A}}$, and $\Delta_{\hat{P}_\tau}^{\tilde{A}}$, respectively, and differential of $\tilde{A}$ as $\Delta_{\tilde{A}}$. It can be noticed by taking differentiation of the equations (\ref{2.10}) and (\ref{2.11}) that $\Delta_{\tilde{P}_\tau}^{\tilde{A}}$ and $\Delta_{\hat{P}_\tau}^{\tilde{A}}$ are related to $\Delta_{\tilde{A}}$ as
\begin{align}
\tilde{A}\Delta_{\tilde{P}_\tau}^{\tilde{A}}+\Delta_{\tilde{P}_\tau}^{\tilde{A}}\tilde{A}^T+\sum_{k=1}^{m}(\tilde{N}_k\Delta_{\tilde{P}_\tau}^{\tilde{A}}\tilde{N}_k^T-e^{\tilde{A}\tau}\tilde{N}_k\Delta_{\tilde{P}_\tau}^{\tilde{A}}\tilde{N}_k^Te^{\tilde{A}^T\tau})+M_1&=0,\label{3.12}\\
A\Delta_{\hat{P}_\tau}^{\tilde{A}}+\Delta_{\tilde{P}_\tau}^{\tilde{A}}\tilde{A}^T+\sum_{k=1}^{m}(N_k\Delta_{\hat{P}_\tau}^{\tilde{A}}\tilde{N}_k^T-e^{A\tau}N_k\Delta_{\hat{P}_\tau}^{\tilde{A}}\tilde{N}_k^Te^{\tilde{A}^T\tau})+M_2&=0\label{3.13}
\end{align}
where
\begin{align}
M_1&=\Delta_{\tilde{A}}\tilde{P}_\tau+\tilde{P}_\tau\Delta_{\tilde{A}}^T-\tau\sum_{k=1}^{m}(\Delta_{\tilde{A}}e^{\tilde{A}\tau}\tilde{N}_k\tilde{P}_{\tau}\tilde{N}_k^Te^{\tilde{A}^T\tau}+e^{\tilde{A}\tau}\tilde{N}_k\tilde{P}_\tau\tilde{N}_k^Te^{\tilde{A}^T\tau}\Delta_{\tilde{A}}^T)\nonumber\\
&\hspace*{5cm}-\tau(\Delta_{\tilde{A}}e^{\tilde{A}\tau}\tilde{B}\tilde{B}^Te^{\tilde{A}^T\tau}+e^{\tilde{A}\tau}\tilde{B}\tilde{B}^Te^{\tilde{A}^T\tau}\Delta_{\tilde{A}}^T),\nonumber\\
M_2&=\hat{P}_\tau\Delta_{\tilde{A}}^T-\tau\sum_{k=1}^{m}e^{A\tau}N_k\hat{P}_\tau\tilde{N}_k^Te^{\tilde{A}^T\tau}\Delta_{\tilde{A}}^T-\tau e^{A\tau}B\tilde{B}^Te^{\tilde{A}^T\tau}\Delta_{\tilde{A}}^T.\nonumber
\end{align}
Now
\begin{align}
\Delta_{J}^{\tilde{A}}=trace(-2C\Delta_{\hat{P}_\tau}^{\tilde{A}}\tilde{C}^T+\tilde{C}\Delta_{\tilde{P}_\tau}^{\tilde{A}}\tilde{C}^T)=trace(-2\tilde{C}^TC\Delta_{\hat{P}_\tau}^{\tilde{A}}+\tilde{C}^T\tilde{C}\Delta_{\tilde{P}_\tau}^{\tilde{A}}).\nonumber
\end{align}
By applying Lemma \ref{lemma1} on (\ref{3.10}) and (\ref{3.13}), and on (\ref{3.11}) and (\ref{3.12}), we get
\begin{align}
\Delta_{J}^{\tilde{A}}&=trace(2M_2R_\tau^T)+trace(M_1S_\tau)\nonumber\\
\Delta_{J}^{\tilde{A}}&=2trace\Big(\hat{P}_\tau\Delta_{\tilde{A}}^TR_\tau^T-\tau\sum_{k=1}^{m}e^{A\tau}N_k\hat{P}_\tau\tilde{N}_k^Te^{\tilde{A}^T\tau}\Delta_{\tilde{A}}^TR_\tau^T-\tau e^{A\tau}B\tilde{B}^Te^{\tilde{A}^T\tau}\Delta_{\tilde{A}}^TR_\tau^T\Big)\nonumber\\
&+trace\Big(\Delta_{\tilde{A}}\tilde{P}_\tau S_\tau+\tilde{P}_\tau\Delta_{\tilde{A}}^TS_\tau-\tau\sum_{k=1}^{m}(\Delta_{\tilde{A}}e^{\tilde{A}\tau}\tilde{N}_k\tilde{P}_{\tau}\tilde{N}_k^Te^{\tilde{A}^T\tau}S_\tau\nonumber\\
&+e^{\tilde{A}\tau}\tilde{N}_k\tilde{P}_\tau\tilde{N}_k^Te^{\tilde{A}^T\tau}\Delta_{\tilde{A}}^TS_\tau)-\tau(\Delta_{\tilde{A}}e^{\tilde{A}\tau}\tilde{B}\tilde{B}^Te^{\tilde{A}^T\tau}S_\tau+e^{\tilde{A}\tau}\tilde{B}\tilde{B}^Te^{\tilde{A}^T\tau}\Delta_{\tilde{A}}^TS_\tau)\Big)\nonumber\\
\Delta_{J}^{\tilde{A}}&=2trace\Big(\Delta_{\tilde{A}}^TR_\tau^T\hat{P}_\tau-\tau\sum_{k=1}^{m}\Delta_{\tilde{A}}^TR_\tau^Te^{A\tau}N_k\hat{P}_\tau\tilde{N}_k^Te^{\tilde{A}^T\tau}-\tau\Delta_{\tilde{A}}^TR_\tau^T e^{A\tau}B\tilde{B}^Te^{\tilde{A}^T\tau}\Big)\nonumber\\
&+trace\Big(\Delta_{\tilde{A}}^TS_\tau\tilde{P}_\tau +\Delta_{\tilde{A}}^TS_\tau\tilde{P}_\tau-\tau\sum_{k=1}^{m}(\Delta_{\tilde{A}}^TS_\tau e^{\tilde{A}\tau}\tilde{N}_k \tilde{P}_{\tau}\tilde{N}_k^Te^{\tilde{A}^T\tau}\nonumber\\
&+\Delta_{\tilde{A}}^TS_\tau e^{\tilde{A}\tau}\tilde{N}_k\tilde{P}_\tau\tilde{N}_k^Te^{\tilde{A}^T\tau})-\tau(\Delta_{\tilde{A}}^TS_\tau e^{\tilde{A}\tau} \tilde{B}\tilde{B}^T e^{\tilde{A}^T\tau}+\Delta_{\tilde{A}}^TS_\tau e^{\tilde{A}\tau}\tilde{B}\tilde{B}^Te^{\tilde{A}^T\tau})\Big)\nonumber\\
\Delta_{J}^{\tilde{A}}&=trace\Big(2\big(R_\tau^T\hat{P}_\tau+S_\tau\tilde{P}_\tau-Y_\tau\big)^T\Delta_{\tilde{A}})\Big).\nonumber\end{align}
Since $\Delta_{J}^{\tilde{A}}=trace\big((\frac{\partial J}{\partial \tilde{A}})^T\Delta_{\tilde{A}}\big)$ \citep{xu2017approach}, $\frac{\partial J}{\partial \tilde{A}}=2(R_\tau^T\hat{P}_\tau+S_\tau\tilde{P}_\tau-Y_\tau)$.

Now suppose that the first-order derivative of $J$, $\tilde{P}_\tau$, and $\hat{P}_\tau$ with respect to $\tilde{N}_k$ is represented by $\Delta_{J}^{\tilde{N}_k}$, $\Delta_{\tilde{P}_\tau}^{\tilde{N}_k}$, and $\Delta_{\hat{P}_\tau}^{\tilde{N}_k}$, respectively, and differential of $\tilde{N}_k$ by $\Delta_{\tilde{N}_k}$. Notice that by taking differentiation of (\ref{2.10}) and (\ref{2.11}), $\Delta_{\tilde{P}_\tau}^{\tilde{N}_k}$ and $\Delta_{\hat{P}_\tau}^{\tilde{N}_k}$ can be related to $\Delta_{\tilde{N}_k}$ by using
\begin{align}
\tilde{A}\Delta_{\tilde{P}_\tau}^{\tilde{N}_k}+\Delta_{\tilde{P}_\tau}^{\tilde{N}_k}\tilde{A}^T+\sum_{k=1}^{m}(\tilde{N}_k\Delta_{\tilde{P}_\tau}^{\tilde{N}_k}\tilde{N}_k^T-e^{\tilde{A}\tau}\tilde{N}_k\Delta_{\tilde{P}_\tau}^{\tilde{N}_k}\tilde{N}_k^Te^{\tilde{A}^T\tau})+M_3&=0,\label{3.14}\\
A\Delta_{\hat{P}_\tau}^{\tilde{N}_k}+\Delta_{\hat{P}_\tau}^{\tilde{N}_k}\tilde{A}^T+\sum_{k=1}^{m}(N_k\Delta_{\hat{P}_\tau}^{\tilde{N}_k}\tilde{N}_k^T-e^{A\tau}N_k\Delta_{\hat{P}_\tau}^{\tilde{N}_k}\tilde{N}_k^Te^{\tilde{A}^T\tau})+M_4&=0\label{3.15}
\end{align}
where
\begin{align}
M_3&=\sum_{k=1}^{m}(\tilde{N}_k\tilde{P}_\tau\Delta_{\tilde{N}_k}^T+\Delta_{\tilde{N}_k}\tilde{P}_\tau\tilde{N}_k^T-e^{\tilde{A}\tau}\tilde{N}_k\tilde{P}_\tau\Delta_{\tilde{N}_k}^Te^{\tilde{A}^T\tau}-e^{\tilde{A}\tau}\Delta_{\tilde{N}_k}\tilde{P}_\tau\tilde{N}_k^Te^{\tilde{A}^T\tau}),\nonumber\\
M_4&=\sum_{k=1}^{m}(N_k\hat{P}_\tau\Delta_{\tilde{N}_k}^T-e^{A\tau}N_k\hat{P}_\tau\Delta_{\tilde{N}_k}^Te^{\tilde{A}\tau}).\nonumber
\end{align}
Now
\begin{align}
\Delta_{J}^{\tilde{N}_k}=trace(-2C\Delta_{\hat{P}_\tau}^{\tilde{N}_k}\tilde{C}^T+\tilde{C}\Delta_{\tilde{P}_\tau}^{\tilde{N}_k}\tilde{C}^T)=trace(-2\tilde{C}^TC\Delta_{\hat{P}_\tau}^{\tilde{N}_k}+\tilde{C}^T\tilde{C}\Delta_{\tilde{P}_\tau}^{\tilde{N}_k}).\nonumber
\end{align}
By applying Lemma \ref{lemma1} on (\ref{3.10}) and (\ref{3.15}), and on (\ref{3.11}) and (\ref{3.14}), we get
\begin{align}
\Delta_{J}^{\tilde{N}_k}&=trace(2M_4R_\tau^T)+trace(M_3S_\tau)\nonumber\\
\Delta_{J}^{\tilde{N}_k}&=2trace\big(\sum_{k=1}^{m}(N_k\hat{P}_\tau\Delta_{\tilde{N}_k}^TR_\tau^T-e^{A\tau}N_k\hat{P}_\tau\Delta_{\tilde{N}_k}^Te^{\tilde{A}\tau}R_\tau^T)\big)\nonumber\\
&+trace\big(\sum_{k=1}^{m}(\tilde{N}_k\tilde{P}_\tau\Delta_{\tilde{N}_k}^TS_\tau+\Delta_{\tilde{N}_k}\tilde{P}_\tau\tilde{N}_k^TS_\tau-e^{\tilde{A}\tau}\tilde{N}_k\tilde{P}_\tau\Delta_{\tilde{N}_k}^Te^{\tilde{A}^T\tau}S_\tau\nonumber\\
&\hspace*{8cm}-e^{\tilde{A}\tau}\Delta_{\tilde{N}_k}\tilde{P}_\tau\tilde{N}_k^Te^{\tilde{A}^T\tau}S_\tau)\big)\nonumber\\
\Delta_{J}^{\tilde{N}_k}&=2trace\big(\sum_{k=1}^{m}(\Delta_{\tilde{N}_k}^TR_\tau^TN_k\hat{P}_\tau-\Delta_{\tilde{N}_k}^Te^{\tilde{A}\tau}R_\tau^Te^{A\tau}N_k\hat{P}_\tau)\big)\nonumber\\
&+trace\big(\sum_{k=1}^{m}(\Delta_{\tilde{N}_k}^TS_\tau\tilde{N}_k\tilde{P}_\tau+\Delta_{\tilde{N}_k}^TS_\tau\tilde{N}_k\tilde{P}_\tau\nonumber\\
&\hspace*{4cm}-\Delta_{\tilde{N}_k}^Te^{\tilde{A}^T\tau}S_\tau e^{\tilde{A}\tau}\tilde{N}_k\tilde{P}_\tau-\Delta_{\tilde{N}_k}^Te^{\tilde{A}^T\tau}S_\tau e^{\tilde{A}\tau}\tilde{N}_k\tilde{P}_\tau)\big)\nonumber\\
\Delta_{J}^{\tilde{N}_k}&=trace\Big(\big(2\sum_{k=1}^{m}R_\tau^TN_k\hat{P}_\tau+S_\tau\tilde{N}_k\tilde{P}_\tau-Z_\tau\big)^T\Delta_{\tilde{N}_k}\Big)\nonumber
\end{align}
Since $\Delta_{J}^{\tilde{N}_k}=trace\big((\frac{\partial J}{\partial \tilde{N}_k})^T\Delta_{\tilde{N}_k}\big)$, $\frac{\partial J}{\partial \tilde{N}_k}=2(\sum_{k=1}^{m}R_\tau^TN_k\hat{P}_\tau+S_\tau\tilde{N}_k\tilde{P}_\tau-Z_\tau)$.

Also, let us denote the first-order derivative of $J$ with respect to $\tilde{B}$ as $\Delta_{J}^{\tilde{B}}$, and let $\Delta_{\tilde{B}}$ be the differential of $\tilde{B}$. Then
\begin{align}
\Delta_{J}^{\tilde{B}}&=trace(2B^T\hat{Q}_\tau\Delta_{\tilde{B}}+\tilde{B}^T\tilde{Q}_\tau\Delta_{\tilde{B}}+\Delta_{\tilde{B}}^T\tilde{Q}_\tau\tilde{B})\nonumber\\
\Delta_{J}^{\tilde{B}}&=trace\big(2(\hat{Q}_\tau^TB+\tilde{Q}_\tau\tilde{B})^T\Delta_{\tilde{B}}\big).\nonumber
\end{align}
Since $\Delta_{J}^{\tilde{B}}=trace\big((\frac{\partial J}{\partial \tilde{B}})^T\Delta_{\tilde{B}}\big)$, $\frac{\partial J}{\partial \tilde{B}}=2(\hat{Q}_\tau^TB+\tilde{Q}_\tau\tilde{B})$.

Finally, if the first-order derivative of $J$ with respect to $\tilde{C}$ is defined as as $\Delta_{J}^{\tilde{C}}$, and $\Delta_{\tilde{C}}$ is the differential of $\tilde{C}$,
\begin{align}
\Delta_{J}^{\tilde{C}}&=trace(-2C\hat{P}_\tau\Delta_{\tilde{C}}^T+\tilde{C}\tilde{P}_\tau\Delta_{\tilde{C}^T}+\Delta_{\tilde{C}}\tilde{P}_\tau\tilde{C}^T)\nonumber\\
\Delta_{J}^{\tilde{C}}&=trace\big(-2(C\hat{P}_\tau+\tilde{C}\tilde{P}_\tau)^T\Delta_{\tilde{C}}\big).\nonumber
\end{align}
Since $\Delta_{J}^{\tilde{C}}=trace\big((\frac{\partial J}{\partial \tilde{C}})^T\Delta_{\tilde{C}}\big)$, $\frac{\partial J}{\partial \tilde{C}}=2(-C\hat{P}_\tau+\tilde{C}\tilde{P}_\tau)$. This completes the proof.\end{proof}
Thus the first-order optimality conditions for the local optimum of $||\Sigma_e||^2_{\mathcal{H}_{2,\tau}}$ can be defined as
\begin{align}
R_\tau^T\hat{P}_\tau+S_\tau\tilde{P}_\tau &=Y_\tau,\label{2.16}\\
\sum_{k=1}^{m}(R_\tau^TN_k\hat{P}_\tau+S_\tau\tilde{N}_k\tilde{P}_\tau)&=\sum_{k=1}^{m}Z_\tau,\\
\hat{Q}_\tau^TB+\tilde{Q}_\tau\tilde{B}&=0,\label{2.18}\\
C\hat{P}_\tau-\tilde{C}\tilde{P}_\tau&=0.\label{2.19}
\end{align}

Next, we present an algorithm (similar to its frequency domain counterpart, i.e., the FLHMORA) that generates a ROM, which approximately satisfies the optimality conditions (\ref{2.16})-(\ref{2.19}). We refer to Algorithm \ref{alg1} as the time-limited $\mathcal{H}_2$-MOR algorithm (TLHMORA).
\begin{algorithm}[!h]
  \caption{TLHMORA}\label{alg1}
\textbf{\textit{Input:}} Original system : $(A,N_k,B,C)$; initial guess: $(\bar{A},\bar{N}_k,\bar{B},\bar{C})$; desired time interval: $[0,\tau]$ sec.\\
\textbf{\textit{Output:}} ROM: $(\tilde{A},\tilde{N}_k,\tilde{B},\tilde{C})$.
  \begin{algorithmic}[1]
\STATE \textbf{while} (not converged) \textbf{do}
\STATE Solve \begin{align}AV_\tau+V_\tau\bar{A}^T+\sum_{k=1}^{m}(N_kV_\tau\bar{N}_k^T-e^{A\tau}N_kV_\tau\bar{N}_k^Te^{\bar{A}^T\tau})+B\bar{B}^T-e^{A\tau}B\bar{B}^Te^{\bar{A}^T\tau}=0.\nonumber\end{align}
\STATE Solve \begin{align}A^TW_\tau+W_\tau\bar{A}+\sum_{k=1}^{m}(N_k^TW_\tau\bar{N}_k-e^{A^T\tau}N_k^TW_\tau\bar{N}_ke^{\bar{A}\tau})-C^T\bar{C}+e^{A^T\tau}C^T\bar{C}e^{\bar{A}\tau}=0.\nonumber\end{align}
\STATE $V=orth(V_\tau)$, $W=orth(W_\tau)$, $W=W(V^TW)^{-1}$.
\STATE $\bar{A}=W^TAV$, $\bar{N}_k=W^TN_kV$, $\bar{B}=W^TB$, $\bar{C}=CV$.
\STATE \textbf{end while}
\STATE $\tilde{A}=\bar{A}$, $\tilde{N}_k=\bar{N}_k$, $\tilde{B}=\bar{B}$, $\tilde{C}=\bar{C}$.
  \end{algorithmic}
\end{algorithm}
\begin{remark}
The selection of the initial guess of the ROM, i.e., $(\bar{A},\bar{N}_k,\bar{B},\bar{C})$, is assumed to be arbitrary in Algorithm \ref{alg1}. In case of linear systems, an appropriate choice of the initial guess is to select $\bar{A},\bar{B},\bar{C}$ such that it contains the dominant eigenvalues of $A$ (i.e., eigenvalues with large residues) and their associated residues. This is because these poles have a big contribution to the $\mathcal{H}_2$-norm of the error transfer function \citep{gugercin2008h_2}. Since bilinear systems closely resemble linear systems from a system theory perspective, one possible approach can be to compute the reduction subspaces $V$ and $W$ that spans the dominant eigenspace of $A$ by using the computationally efficient eigensolver proposed in \citep{rommes2006efficient}. $V$ and $W$ can then be used to generate the initial guess. Another possible approach (again inspired by the linear case \citep{vuillemin2013h2}) is to start the HOMORA \citep{breiten2012interpolation} arbitrarily and then use its final ROM as an initial guess in Algorithm \ref{alg1}.
\end{remark}
\begin{remark}
When $N_k=0$, the optimality conditions (\ref{2.16})-(\ref{2.19}) reduce to the one in the linear case \citep{goyal2019time}. Accordingly, the TLHMORA reduces to the algorithm presented in \citep{goyal2019time}.
\end{remark}
\begin{remark}
When the desired time interval is set to $[0,\infty]$, $Y_\tau=0$, $Z_\tau=0$, and the TLHMORA reduces to the HOMORA.
\end{remark}
\subsection{Time- and Frequency-limited $\mathcal{H}_{2}$-pseudo-optimal MOR}
In this subsection, we discuss the reason why the FLHMORA and the TLHMORA may not satisfy their respective optimality conditions despite offering good accuracy. Also, we propose two algorithms that enforce the respective optimality conditions associated with $\tilde{B}$ and $\tilde{C}$ on the ROM.
\subsubsection{Limitation in Projection-type Framework}
The FLHMORA mimics the HOMORA in trying to ensure that $V=\hat{P}_\omega\tilde{P}_\omega^{-1}$ and $W=-\hat{Q}_\omega\tilde{Q}_\omega^{-1}$ upon convergence. Similarly, the TLHMORA mimics the HOMORA in trying to ensure that $V=\hat{P}_\tau\tilde{P}_\tau^{-1}$ and $W=-\hat{Q}_\tau\tilde{Q}_\tau^{-1}$ upon convergence. From the perspective of projection problem, $\tilde{P}_\omega$, $\tilde{Q}_\omega$, $\tilde{P}_\tau$, and $\tilde{Q}_\tau$ can be seen as normalizations to ensure that $W^TV=I$. The condition $W^TV=I$ implies that $\hat{Q}_\omega^T\hat{P}_\omega+\tilde{Q}_\omega\tilde{P}_\omega=0$ or $\hat{Q}_\tau^T\hat{P}_\tau+\tilde{Q}_\tau\tilde{P}_\tau=0$. However, the dilemma, which the FLHMORA and the TLHMORA face, is that $\hat{Q}_\omega^T\hat{P}_\omega+\tilde{Q}_\omega\tilde{P}_\omega=0$ and $\hat{Q}_\tau^T\hat{P}_\tau+\tilde{Q}_\tau\tilde{P}_\tau=0$ do not correspond to the respective optimality conditions associated with $\tilde{A}$. Similarly, $\sum_{k=1}^{m}\tilde{Q}_\omega\tilde{N}_k\tilde{P}_\omega+\hat{Q}_\omega^TN_k\hat{P}_\omega=0$ and $\sum_{k=1}^{m}\tilde{Q}_\tau\tilde{N}_k\tilde{P}_\tau+\hat{Q}_\tau^TN_k\hat{P}_\tau=0$ also do not correspond to the respective optimality conditions associated with $\tilde{N}_k$. Thus the FLHMORA and the TLHMORA fail to ensure the optimality conditions associated with their respective problems upon convergence because these are projection-type algorithms. It can also be noted that $V=\hat{P}_\omega\tilde{P}_\omega^{-1}$ or $V=\hat{P}_\tau\tilde{P}_\tau^{-1}$ and $W=-\hat{Q}_\omega\tilde{Q}_\omega^{-1}$ or $W=-\hat{Q}_\tau\tilde{Q}_\tau^{-1}$ do correspond to the optimality conditions associated with $\tilde{C}$ and $\tilde{B}$, respectively. This is the main reason why these algorithms are able to generate a high-fidelity ROM. In the sequel, we focus on achieving the optimality conditions associated with $\tilde{B}$ and $\tilde{C}$, and we keep $\tilde{A}$ and $\tilde{N}_k$ fixed. The matrices $\tilde{A}$ and $\tilde{N}_k$ can be obtained from the final ROM constructed by the FLHMORA or the TLHMORA. This approach can no longer be considered as a projection-type technique as $\tilde{A}$ and $\tilde{N}_k$ are fixed, and $W^TV\neq I$.
\subsubsection{Frequency-limited pseudo-optimal $\mathcal{H}_2$-MOR algorithm (FLPHMORA)}Let $\bar{P}_{\omega}$ be the frequency-limited controllability gramian and  $\bar{Q}_{\omega}$ be the frequency-limited observability gramian of the initial guess $(\bar{A},\bar{N}_k,\bar{B},\bar{C})$. Then $\bar{P}_{\omega}$ and $\bar{Q}_{\omega}$ solve the following generalized Lyapunov equations
\begin{align}
\bar{A}\bar{P}_\omega+\bar{P}_\omega \bar{A}^T+\sum_{k=1}^{m}\big(F_\omega[\bar{A}]\bar{N}_{k}\bar{P}_\omega \bar{N}_{k}^T&+\bar{N}_{k}\bar{P}_\omega \bar{N}_{k}^TF_\omega[\bar{A}]^T\big)\nonumber\\
&+F_\omega [\bar{A}]\bar{B}\bar{B}^T+\bar{B}\bar{B}^TF_\omega[\bar{A}]^T=0,\label{2.1}\\
\bar{A}^T\bar{Q}_\omega+\bar{Q}_\omega \bar{A}+\sum_{k=1}^{m}\big(F_\omega[\bar{A}]^T\bar{N}_{k}^T\bar{Q}_\omega \bar{N}_{k}&+\bar{N}_{k}^T\bar{Q}_\omega \bar{N}_{k}F_\omega[\bar{A}]\big)\nonumber\\
&+F_\omega [\bar{A}]^T\bar{C}^T\bar{C}+\bar{C}^T\bar{C}F_\omega[\bar{A}]=0.\label{2.2}
\end{align}
It can be readily verified by inspection that the optimality condition (\ref{1.47}) can be enforced in a single-run if $\tilde{\Sigma}$ is computed as the following
\begin{align}
\tilde{A}=\bar{A},\hspace*{2mm}\tilde{N}_k=\bar{N}_k,\hspace*{2mm}\tilde{B}=-\bar{Q}_\omega^{-1}W_\omega^TB,\hspace*{2mm}\tilde{C}=\bar{C}.
\end{align}
It can also be noted that when $\tilde{\Sigma}$ satisfies the optimality condition (\ref{1.47}), the following holds
\begin{align}
||\Sigma_e||_{\mathcal{H}_{2,\omega}}^2&=trace(B^TQ_\omega B-\tilde{B}^T\bar{Q}_\omega\tilde{B})\nonumber\\
&=trace(B^TQ_\omega B-B^TW_\omega\bar{Q}_\omega^{-1}W_\omega^TB)\nonumber\\
&=||\Sigma||_{\mathcal{H}_{2,\omega}}^2-||\tilde{\Sigma}||_{\mathcal{H}_{2,\omega}}^2.\nonumber
\end{align}
Thus $W_\omega\bar{Q}_\omega^{-1}W_\omega^T$ is an approximation of $Q_\omega$. Similarly, the optimality condition (\ref{1.48}) can be enforced in a single-run if $\tilde{\Sigma}$ is computed as the following
\begin{align}
\tilde{A}=\bar{A},\hspace*{2mm}\tilde{N}_k=\bar{N}_k,\hspace*{2mm}\tilde{B}=\bar{B},\hspace*{2mm}\tilde{C}=CV_\omega\bar{P}_\omega^{-1}.
\end{align}
Again, it can readily be noted that when $\tilde{\Sigma}$ satisfies the optimality condition (\ref{1.48}), the following holds
\begin{align}
||\Sigma_e||_{\mathcal{H}_{2,\omega}}^2&=trace(CP_\omega C^T-\tilde{C}\bar{P}_\omega\tilde{C}^T)\nonumber\\
&=trace(CP_\omega C^T-CV_\omega\bar{P}_\omega^{-1}V_\omega^TC^T)\nonumber\\
&=||\Sigma||_{\mathcal{H}_{2,\omega}}^2-||\tilde{\Sigma}||_{\mathcal{H}_{2,\omega}}^2.\nonumber
\end{align}
Thus $V_\omega\bar{P}_\omega^{-1}V_\omega^T$ is an approximation of $P_\omega$.

We now present an algorithm that generates a ROM, which satisfies both the optimality conditions (\ref{1.47}) and (\ref{1.48}) upon convergence. The pseudo-code of our approach is given below in Algorithm \ref{alg2}. The steps (3)-(5) select $\bar{B}$ to enforce the optimality condition (\ref{1.47}). The steps (6)-(8) select $\bar{C}$ to enforce the optimality condition (\ref{1.48}). Thus Algorithm \ref{alg2} enforces (\ref{1.47}) and (\ref{1.48}) upon convergence while keeping $\bar{A}$ and $\bar{N}_k$ fixed.
\begin{algorithm}[!h]
  \caption{FLPHMORA}\label{alg2}
\textbf{\textit{Input:}} Original system : $(A,N_k,B,C)$; desired frequency interval: $[0,\omega]$ rad/sec.\\
\textbf{\textit{Output:}} ROM: $(\tilde{A},\tilde{N}_k,\tilde{B},\tilde{C})$.
\begin{algorithmic}[1]
\STATE Run FLHMORA, and set $\bar{A}=\tilde{A}$ and $\bar{N}_k=\tilde{N}_k$.
\STATE \textbf{while}(not converged) do
\STATE Compute $\bar{Q}_\omega$ from (\ref{2.2}).
\STATE Solve \begin{align}A^TW_\omega+W_\omega\bar{A}+\sum_{k=1}^{m}(F_\omega[A]^TN_k^TW_\omega\bar{N}_k+N_k^TW_\omega\bar{N}_kF_\omega[\bar{A}])-F_\omega[A]^TC^T\bar{C}-C^T\bar{C}F_\omega[\bar{A}]=0.\nonumber\end{align}
\STATE $\bar{B}=-\bar{Q}_\omega^{-1}W_\omega^TB$. $\%$ \textit{Satisfy the optimality condition (\ref{1.47})}.
\STATE Compute $\bar{P}_\omega$ from (\ref{2.1}).
\STATE Solve
\begin{align}AV_\omega+V_\omega\bar{A}^T+\sum_{k=1}^{m}(F_\omega[A]N_kV_\omega\bar{N}_k^T+N_kV_\omega\bar{N}_k^TF_\omega[\bar{A}]^T)+F_\omega[A]B\bar{B}^T+B\bar{B}^TF_\omega[\bar{A}]^T=0.\nonumber\end{align}
\STATE $\bar{C}=CV_\omega\bar{P}_\omega^{-1}$. $\%$ \textit{Satisfy the optimality condition (\ref{1.48})}.
\STATE \textbf{end while}
\STATE $\tilde{A}=\bar{A}$, $\tilde{N}_k=\bar{N}_k$, $\tilde{B}=\bar{B}$, $\tilde{C}=\bar{C}$.
\end{algorithmic}
\end{algorithm}
\begin{remark}
The FLPHMORA also provides the approximations of $P_\omega$ and $Q_\omega$ upon convergence, i.e., $V_\omega\bar{P}_\omega^{-1} V_\omega^T$ and $W_\omega\bar{Q}_\omega^{-1} W_\omega^T$, respectively. These approximations can be used to save some computational cost in FLBT \citep{shaker2013frequency}.
\end{remark}
\subsubsection{Time-limited pseudo-optimal $\mathcal{H}_2$-MOR algorithm (TLPHMORA)}
Let $\bar{P}_{\tau}$ be the time-limited controllability gramian and  $\bar{Q}_{\tau}$ be the time-limited observability gramian of the initial guess $(\bar{A},\bar{N}_k,\bar{B},\bar{C})$. Then $\bar{P}_{\tau}$ and $\bar{Q}_{\tau}$ solve the following generalized Lyapunov equations
\begin{align}
\bar{A}\bar{P}_\tau+\bar{P}_\tau \bar{A}^T+\sum_{k=1}^{m}\big(\bar{N}_{k}\bar{P}_\tau \bar{N}_{k}^T-e^{\bar{A}\tau}\bar{N}_{k}\bar{P}_\tau \bar{N}_{k}^Te^{\bar{A}^T\tau}\big)+\bar{B}\bar{B}^T-e^{\bar{A}\tau}\bar{B}\bar{B}^Te^{\bar{A}^T\tau}=0,\label{2.20}\\
\bar{A}^T\bar{Q}_\tau+\bar{Q}_\tau \bar{A}+\sum_{k=1}^{m}\big(\bar{N}_{k}^T\bar{Q}_\tau \bar{N}_{k}-e^{\bar{A}^T\tau}\bar{N}_{k}^T\bar{Q}_\tau \bar{N}_{k}e^{\bar{A}\tau}\big)+\bar{C}^T\bar{C}-e^{\bar{A}^T\tau}\bar{C}^T\bar{C}e^{\bar{A}\tau}=0.\label{2.21}
\end{align}
It can be verified by inspection that the optimality condition (\ref{2.18}) can be enforced in a single-run if $\tilde{\Sigma}$ is computed as
\begin{align}
\tilde{A}=\bar{A},\hspace*{2mm}\tilde{N}_k=\bar{N}_k,\hspace*{2mm}\tilde{B}=-\bar{Q}_\tau^{-1}W_\tau^TB,\hspace*{2mm}\tilde{C}=\bar{C}.
\end{align}
Note that when $\tilde{\Sigma}$ satisfies the optimality condition (\ref{2.18}), the following holds
\begin{align}
||\Sigma_e||_{\mathcal{H}_{2,\tau}}^2&=trace(B^TQ_\tau B-\tilde{B}^T\bar{Q}_\tau\tilde{B})\nonumber\\
&=trace(B^TQ_\tau B-B^TW_\tau\bar{Q}_\tau^{-1}W_\tau^TB)\nonumber\\
&=||\Sigma||_{\mathcal{H}_{2,\tau}}^2-||\tilde{\Sigma}||_{\mathcal{H}_{2,\tau}}^2.\nonumber
\end{align}
Thus $W_\tau\bar{Q}_\tau^{-1}W_\tau^T$ is an approximation of $Q_\tau$. Similarly, the optimality condition (\ref{2.19}) can be enforced in a single-run if $\tilde{\Sigma}$ is computed as the following
\begin{align}
\tilde{A}=\bar{A},\hspace*{2mm}\tilde{N}_k=\bar{N}_k,\hspace*{2mm}\tilde{B}=\bar{B},\hspace*{2mm}\tilde{C}=CV_\tau\bar{P}_\tau^{-1}.
\end{align}
Again, it can be noted that when $\tilde{\Sigma}$ satisfies the optimality condition (\ref{2.19}), the following holds
\begin{align}
||\Sigma_e||_{\mathcal{H}_{2,\tau}}^2&=trace(CP_\tau C^T-\tilde{C}\bar{P}_\tau\tilde{C}^T)\nonumber\\
&=trace(CP_\tau C^T-CV_\tau\bar{P}_\tau^{-1}V_\tau^TC^T)\nonumber\\
&=||\Sigma||_{\mathcal{H}_{2,\tau}}^2-||\tilde{\Sigma}||_{\mathcal{H}_{2,\tau}}^2.\nonumber
\end{align}
Thus $V_\tau\bar{P}_\tau^{-1}V_\tau^T$ is an approximation of $P_\tau$.

We now present an algorithm that generates a ROM, which satisfies both the optimality conditions (\ref{2.18}) and (\ref{2.19}) upon convergence. The pseudo-code of the proposed approach is given below in Algorithm \ref{alg3}. The steps (3)-(5) selects $\bar{B}$ to enforce the optimality condition (\ref{2.18}). The steps (6)-(8) selects $\bar{C}$ to enforce the optimality condition (\ref{2.19}). Thus Algorithm \ref{alg3} enforces (\ref{2.18}) and (\ref{2.19}) upon convergence while keeping $\bar{A}$ and $\bar{N}_k$ fixed.
\begin{algorithm}[!h]
  \caption{TLPOHMORA}\label{alg3}
\textbf{\textit{Input:}} Original system : $(A,N_k,B,C)$; desired time interval: $[0,\tau]$ sec.\\
\textbf{\textit{Output:}} ROM: $(\tilde{A},\tilde{N}_k,\tilde{B},\tilde{C})$.
\begin{algorithmic}[1]
\STATE Run TLHMORA, and set $\bar{A}=\tilde{A}$ and $\bar{N}_k=\tilde{N}_k$.
\STATE \textbf{while}(not converged) \textbf{do}
\STATE Compute $\bar{Q}_\tau$ from (\ref{2.21}).
\STATE Solve \begin{align}A^TW_\tau+W_\tau\bar{A}+\sum_{k=1}^{m}(N_k^TW_\tau\bar{N}_k-e^{A^T\tau}N_k^TW_\tau\bar{N}_ke^{\bar{A}\tau})-C^T\bar{C}+e^{A^T\tau}C^T\bar{C}e^{\bar{A}\tau}=0.\nonumber\end{align}
\STATE $\bar{B}=-\bar{Q}_\tau^{-1}W_\tau^TB$. $\%$ \textit{Satisfy the optimality condition (\ref{2.18})}.
\STATE Compute $\bar{P}_\tau$ from (\ref{2.20}).
\STATE Solve
\begin{align}AV_\tau+V_\tau\bar{A}^T+\sum_{k=1}^{m}(N_kV_\tau\bar{N}_k^T-e{A\tau}N_kV_\tau\bar{N}_k^Te^{\bar{A}^T\tau})+B\bar{B}^T-e^{A\tau}B\bar{B}^Te^{\bar{A}^T\tau}=0.\nonumber\end{align}
\STATE $\bar{C}=CV_\tau\bar{P}_\tau^{-1}$. $\%$ \textit{Satisfy the optimality condition (\ref{2.19})}.
\STATE \textbf{end while}
\STATE $\tilde{A}=\bar{A}$, $\tilde{N}_k=\bar{N}_k$, $\tilde{B}=\bar{B}$, $\tilde{C}=\bar{C}$.
\end{algorithmic}
\end{algorithm}
\begin{remark}
The TLPHMORA provides the approximations of $P_\tau$ and $Q_\tau$ upon convergence, i.e., $V_\tau\bar{P}_\tau^{-1} V_\tau^T$ and $W_\tau\bar{Q}_\tau^{-1} W_\tau^T$, respectively. These approximations can be used to save some computational cost in TLBT \citep{shaker2014time}.
\end{remark}
\begin{remark}
Throughout the text, we have considered the desired time and frequency intervals as $[0,\tau]$ sec and $[0,\omega]$ rad/sec, respectively. However, the results presented can be generalized for any time and frequency intervals, i.e., $[\tau_1,\tau_2]$ sec and $[\omega_1,\omega_2]$ rad/sec, respectively. For a generic frequency interval $[\omega_1,\omega_2]$ rad/sec, $F_\omega[A]$ becomes $F_\omega[A]=Real\Big(\frac{j}{2\pi}ln\big((j\omega_1I+A)^{-1}(j\omega_2I+A)\big)\Big)$ \citep{petersson2014model}. For a generic time interval $[\tau_1,\tau_2]$ sec, there are more changes. $\bar{P}_\tau$ and $\bar{Q}_\tau$, in this case, solve the following generalized Lyapunov equations \citep{shaker2014time}
\begin{align}
\bar{A}\bar{P}_\tau+\bar{P}_\tau \bar{A}^T+\sum_{k=1}^{m}\big(e^{\bar{A}\tau_1}\bar{N}_{k}\bar{P}_\tau \bar{N}_{k}^Te^{\bar{A}^T\tau_1}&-e^{\bar{A}\tau_2}\bar{N}_{k}\bar{P}_\tau \bar{N}_{k}^Te^{\bar{A}^T\tau_2}\big)\nonumber\\
&+e^{\bar{A}\tau_1}\bar{B}\bar{B}^Te^{\bar{A}^T\tau_1}-e^{\bar{A}\tau_2}\bar{B}\bar{B}^Te^{\bar{A}^T\tau_2}=0,\nonumber\\
\bar{A}^T\bar{Q}_\tau+\bar{Q}_\tau \bar{A}+\sum_{k=1}^{m}\big(e^{\bar{A}^T\tau_1}\bar{N}_{k}^T\bar{Q}_\tau \bar{N}_{k}e^{\bar{A}\tau_1}&-e^{\bar{A}^T\tau_2}\bar{N}_{k}^T\bar{Q}_\tau \bar{N}_{k}e^{\bar{A}\tau_2}\big)\nonumber\\
&+e^{\bar{A}^T\tau_1}\bar{C}^T\bar{C}e^{\bar{A}\tau_1}-e^{\bar{A}^T\tau_2}\bar{C}^T\bar{C}e^{\bar{A}\tau_2}=0.\nonumber
\end{align}
Moreover, $V_\tau$ and $W_\tau$ now solve the following generalized Sylvester equations
\begin{align}
AV_\tau+V_\tau\bar{A}^T+\sum_{k=1}^{m}(e^{A\tau_1}N_kV_\tau\bar{N}_k^Te^{\bar{A}^T\tau_2}&-e^{A\tau_2}N_kV_\tau\bar{N}_k^Te^{\bar{A}^T\tau_2})\nonumber\\
&+e^{A\tau_1}B\bar{B}^Te^{\bar{A}^T\tau_1}-e^{A\tau_2}B\bar{B}^Te^{\bar{A}^T\tau_2}=0,\nonumber\\ A^TW_\tau+W_\tau\bar{A}+\sum_{k=1}^{m}(e^{A^T\tau_1}N_k^TW_\tau\bar{N}_ke^{\bar{A}\tau_1}&-e^{A^T\tau_2}N_k^TW_\tau\bar{N}_ke^{\bar{A}\tau_2})\nonumber\\
&-e^{A^T\tau_1}C^T\bar{C}e^{\bar{A}\tau_1}+e^{A^T\tau_2}C^T\bar{C}e^{\bar{A}\tau_2}=0.\nonumber
\end{align}
\end{remark}
\begin{remark}
The selection of the order $r$ of ROM cannot be made on the fly as no computable \textit{apriori} error bound expression is available. Thus the value of $r$ can be increased to obtain a more accurate ROM in case the error $||\Sigma_e||_{\mathcal{H}_{2,\tau}}$ or $||\Sigma_e||_{\mathcal{H}_{2,\omega}}$ is greater than the desired tolerance.
\end{remark}
\begin{remark}
The stability of the ROM can not be guaranteed theoretically by all the algorithms under consideration. But since the stability of the ROM depends on the matrices $\tilde{A}$ and $\tilde{N}_k$, Algorithms \ref{alg2} and \ref{alg3} can ensure that the final ROM is stable in an \textit{ad-hoc} sense. This can be done by rejecting any choice of $\bar{A}$ and $\bar{N}_k$ in step $1$ if $\bar{A}$ has eigenvalues in the right half of the $s$-plane or $\bar{N}_k$ is unbounded. See \citep{breiten2012interpolation} for a precise definition of a stable bilinear system.
\end{remark}
\subsection{Computational Aspects}
The TLBT, the TLHMORA, and the TLPOHMORA require the computation of $e^{A\tau}B$, $e^{A^T\tau}C^T$, and $e^{A\tau}N_k$, which is expensive in large-scale setting. These products of matrix exponentials need to be approximated in case $\Sigma$ is large-scale. Similarly, the FLBT, the FLHMORA, and the FLPOHMORA require the computation of $F_\omega[A]B$, $F_\omega[A]^T C^T$, and $F_\omega[A]N_k$, which is expensive in large-scale setting. As shown in \citep{kurschner2018balanced} and \citep{benner2016frequency}, $e^{A\tau}B$, $e^{A^T\tau}C^T$, $F_\omega[A]B$, and $F_\omega[A]^T C^T$ can be approximated by using projection (when $\Sigma$ is a linear system) as
\begin{align}
e^{A\tau}B&\approx Ve^{\tilde{A}\tau}\tilde{B},&e^{A^T\tau}C^T&\approx We^{\tilde{A}^T\tau}\tilde{C}^T,\\
F_\omega[A]B&\approx VF_\omega[\tilde{A}]\tilde{B},&F_\omega[A]^TC^T&\approx WF_\omega[\tilde{A}]^T\tilde{C}^T.
\end{align}
Similarly, $e^{A\tau}N_k$ and $F_\omega[A]N_k$ can be approximated in bilinear setting by using
\begin{align}
e^{A\tau}N_k&\approx Ve^{\tilde{A}\tau}\tilde{N}_kW^T,&&\textnormal{and}&F_\omega[A]N_k&\approx VF_\omega[\tilde{A}]\tilde{N}_kW^T.
\end{align}
The reduction subspaces $V$ and $W$ for this purpose can be generated by using HOMORA as it accurately captures the dynamics of $\Sigma$. The iterative algorithm can be stopped when the relative change in $e^{\tilde{A}\tau}\tilde{B}$, $e^{\tilde{A}^T\tau}\tilde{C}^T$, $e^{\tilde{A}\tau}\tilde{N}_k$, $F_\omega[\tilde{A}]\tilde{B}$, $F_\omega[\tilde{A}]^T\tilde{C}^T$, and $F_\omega[\tilde{A}]\tilde{N}_k$ stagnate as we are not interested in constructing an $\mathcal{H}_2$-optimal ROM. This approach is summarized in Algorithm \ref{alg4}.
\begin{algorithm}[!h]
  \caption{Approximation of exponential and logarithmic products}\label{alg4}
\textbf{\textit{Input:}} Original system : $(A,N_k,B,C)$ and initial guess of $(\tilde{A},\tilde{N}_k,\tilde{B},\tilde{C})$\\
\textbf{\textit{Output:}} Approximation of $e^{A\tau}B$, $e^{A^T\tau}C^T$, $e^{A\tau}N_k$, $F_\omega[A]B$, $F_\omega[A]^TC^T$, and $F_\omega[A]N_k$.
\begin{algorithmic}[1]
\STATE \textbf{while}(not converged) \textbf{do}
\STATE Compute $AV+V\tilde{A}^T+\sum_{k=1}^{m}N_kV\tilde{N}_k^T+B\tilde{B}^T=0$.
\STATE Compute $A^TW+W\tilde{A}+\sum_{k=1}^{m}N_k^TW\tilde{N}_k-C^T\tilde{C}=0$.
\STATE $V=orth(V)$, $W=orth(W)$, $W=W(V^TW)^{-1}$.
\STATE $\tilde{A}=W^TAV$, $\tilde{N}_k=W^TN_kV$, $\tilde{B}=W^TB$, $\tilde{C}=CV$.
\STATE $e^{A\tau}B\approx Ve^{\tilde{A}\tau}\tilde{B}$, $e^{A^T\tau}C^T\approx We^{\tilde{A}^T\tau}\tilde{C}^T$, $e^{A\tau}N_k\approx Ve^{\tilde{A}\tau}\tilde{N}_kW^T$.
\STATE $F_\omega[A]B\approx VF_\omega[\tilde{A}]\tilde{B}$, $F_\omega[A]^TC^T\approx WF_\omega[\tilde{A}]^T\tilde{C}^T$, $F_\omega[A]N_k\approx VF_\omega[\tilde{A}]\tilde{N}_kW^T$.
\STATE \textbf{end while}
\end{algorithmic}
\end{algorithm}
\begin{remark}
The solution generalized Sylvester equations for computing $V_\tau$ and $W_\tau$ is a computational challenge for large-scale $A$. Note that by using vectorization operator $vec(\cdot)$ and Kronecker product, $V_\tau$ and $W_\tau$ can be computed by solving the following linear system of equations
\begin{align}
\Big(I\otimes A+\bar{A}\otimes I+\sum_{k=1}^{m}e^{\bar{A}\tau_1}\bar{N}_k\otimes &e^{A\tau_1}N_k-\sum_{k=1}^{m}e^{\bar{A}\tau_2}\bar{N}_k\otimes e^{A\tau_2}N_k\Big)vec(V_\tau)\nonumber\\
&=-vec\big(e^{A\tau_1}B\bar{B}^Te^{\bar{A}^T\tau_1}-e^{A\tau_2}B\bar{B}^Te^{\bar{A}^T\tau_2}\big),\nonumber\\
\Big(I\otimes A^T+\bar{A}^T\otimes I+\sum_{k=1}^{m}e^{\bar{A^T}\tau_1}\bar{N}_k^T\otimes &e^{A^T\tau_1}N_k^T-\sum_{k=1}^{m}e^{\bar{A}^T\tau_2}\bar{N}_k^T\otimes e^{A^T\tau_2}N_k^T\Big)vec(W_\tau)\nonumber\\
&=vec\big(e^{A^T\tau_1}C^T\bar{C}e^{\bar{A}\tau_1}-e^{A^T\tau_2}C^T\bar{C}e^{\bar{A}\tau_2}\big).\nonumber
\end{align}
The approximate solutions of the above linear systems of equations can be obtained within admissible time with the iterative Krylov subspace based solvers proposed in \citep{bouhamidi2008note}, and by using the algorithms in \citep{bennersparse} for preconditioning.
\end{remark}
\begin{remark}
$V_\tau$ and $W_\tau$ can also be computed iteratively \citep{breiten2012interpolation,shaker2014time} by using
\begin{align}
AV_{1,\tau}&+V_{1,\tau}\bar{A}^T+e^{A\tau_1}B\bar{B}^Te^{\bar{A}^T\tau_1}-e^{A\tau_2}B\bar{B}^Te^{\bar{A}^T\tau_2}=0,\nonumber\\
AV_{i,\tau}&+V_{i,\tau}\bar{A}^T+\sum_{k=1}^{m}(e^{A\tau_1}N_kV_{i,\tau}\bar{N}_k^Te^{\bar{A}^T\tau_1}-e^{A\tau_2}N_kV_{i,\tau}\bar{N}_k^Te^{\bar{A}^T\tau_2})\nonumber\\
&\hspace*{5cm}+e^{A\tau_1}B\bar{B}^Te^{\bar{A}^T\tau_1}-e^{A\tau_2}B\bar{B}^Te^{\bar{A}^T\tau_2}=0,\nonumber\\
A^TW_{1,\tau}&+W_{1,\tau}\bar{A}-e^{A^T\tau_1}C^T\bar{C}e^{\bar{A}\tau_1}+e^{A^T\tau_2}C^T\bar{C}e^{\bar{A}\tau_2}=0,\nonumber\\
A^TW_{i,\tau}&+W_{i,\tau}\bar{A}-\sum_{k=1}^{m}(e^{A^T\tau_1}N_k^TW_{i,\tau}\bar{N}_ke^{\bar{A}\tau_1}-e^{A^T\tau_2}N_k^TW_{i,\tau}\bar{N}_ke^{\bar{A}\tau_2})\nonumber\\
&\hspace*{5cm}-e^{A^T\tau_1}C^T\bar{C}e^{\bar{A}\tau_1}+e^{A^T\tau_2}C^T\bar{C}e^{\bar{A}\tau_2}=0,\nonumber
\end{align} where $V_\tau=\lim_{i \to \infty}V_{i,\tau}$ and $W_\tau=\lim_{i \to \infty}W_{i,\tau}$. The computational cost in this approach can be controlled by truncating the iterations at a small value of $i$ as done in \citep{benner2017truncated}.
\end{remark}
\begin{remark}
It should be stressed here that the computation of $P_\tau$ and $Q_\tau$ in the TLBT \citep{shaker2014time} is more expensive than $V_\tau$ and $W_\tau$ because the equations (\ref{eq:16}) and (\ref{eq:17}) do not involve small-scale matrix $\bar{A}$. Thus the computational superiority of the TLHMORA and the TLPHMORA over TLBT is obvious, provided they converge quickly.
\end{remark}
\section{Numerical Examples}
In this section, we test our algorithms on three numerical examples. These are standard examples considered in the literature to test MOR algorithms for bilinear control systems \citep{benner2011lyapunov,breiten2012interpolation,shaker2013frequency,ahmad2017implicit,xu2017approach}. We use the ROMs generated by the BT, the FLBT, and the TLBT to initialize the HOMORA, the FLHMORA, and the TLHMORA, respectively. We obtain the solutions of generalized Sylvester and Lyapunov equations using the iterative method described in the last section (and also in \citep{breiten2012interpolation,shaker2013frequency,shaker2014time}). We truncate the solution after $3$ iterations because we do not see any significant change or improvement in the results by going over $3$ iterations. We solve the Lyapunov equations exactly by using MATLAB's \textit{lyap} command for all the algorithms. Similarly, we compute the matrix exponentials and matrix logarithms exactly for all the examples by using MATLAB's \textit{expm} and \textit{logm} commands, respectively. The tolerance for convergence in the HOMORA, the FLHMORA, the TLHMORA, the FLPHMORA, and the TLPHMORA is set to $1\times 10^{-5}$. All the experiments are performed using MATLAB $2016$ on a laptop with $2$GHz $i7$ Intel processor and $16$GB random access memory (RAM).

\textbf{Illustrative Example:} Consider a $7^{th}$ order illustrative example from \citep{shaker2013frequency}, which has the following state-space matrices
\begin{align}
A&=\begin{bmatrix}-0.81 & 0.47 & -0.43 & 1.6 & 0.26 & -0.4 & 0.92\\
    -0.61 & -1.9 & 0.8 & -1.6 & 2 & 0.98 & -0.9\\
    0.5 & -1.2 & -2.1 & -1.6 & -1.1 & 0.14 & -0.87\\
    -1.3 & 2.1 & 0.47 & -1.2 & 3.7 & -1.2 & -1.3\\
    -0.24 & -0.081 & 1.6 & -3.6 & -1.3 & 1.7 & -2.6\\
    1.3 & -0.96 & -1.3 & -0.57 & -2.4 & -2.4 & -0.36\\
    -0.16 & 1.5 & -0.99 & 1.5 & 0.61 & -2.2 & -3.3\end{bmatrix},\nonumber\\
N&=\begin{bmatrix}-1& 0 & 0 & 0 & 0 & 0 & 0\\
0& -1 & 0 & 0 & 0 & 0 & 0\\
0& 0 & -1 & 0 & 0 & 0 & 0\\
0& 0 & 0 & -1 & 0 & 0 & 0\\
0& 0 & 0 & 0 & -1 & 0 & 0\\
0& 0 & 0 & 0 & 0 & -1 & 0\\
0& 0 & 0 & 0 & 0 & 0 & -1\end{bmatrix},\hspace*{2mm}B=\begin{bmatrix}0 \\ 0 \\ -0.196 \\ 1.42 \\ 0.292 \\ 0.198 \\ 1.59\end{bmatrix},\nonumber\\
C&=\begin{bmatrix}-0.804 &0 & 0.835 & -0.244 & 0.216 & -1.17 & -1.15\end{bmatrix}.\nonumber
\end{align} Let the input signal be a sinusoid with the frequency and amplitude of $5$ rad/sec and $0.01$, respectively, i.e., $u(t)=0.01sin(5t)$. We obtain $1^{st}$ order ROMs using the BT, the FLBT, the HOMORA, the FLHMORA, and the FLPHMORA. We set the desired frequency interval as $[4,6]$ rad/sec in the FLBT, the FLHMORA, and the FLPHMORA to ensure good accuracy at and in close neighbourhood of $5$ rad/sec. The absolute error in the output response is compared in Figure \ref{fig1} on a logarithmic scale, and it can be seen that the FLPHMORA provides the best approximation.
\begin{figure}[!h]
\centering
\includegraphics[width=12cm]{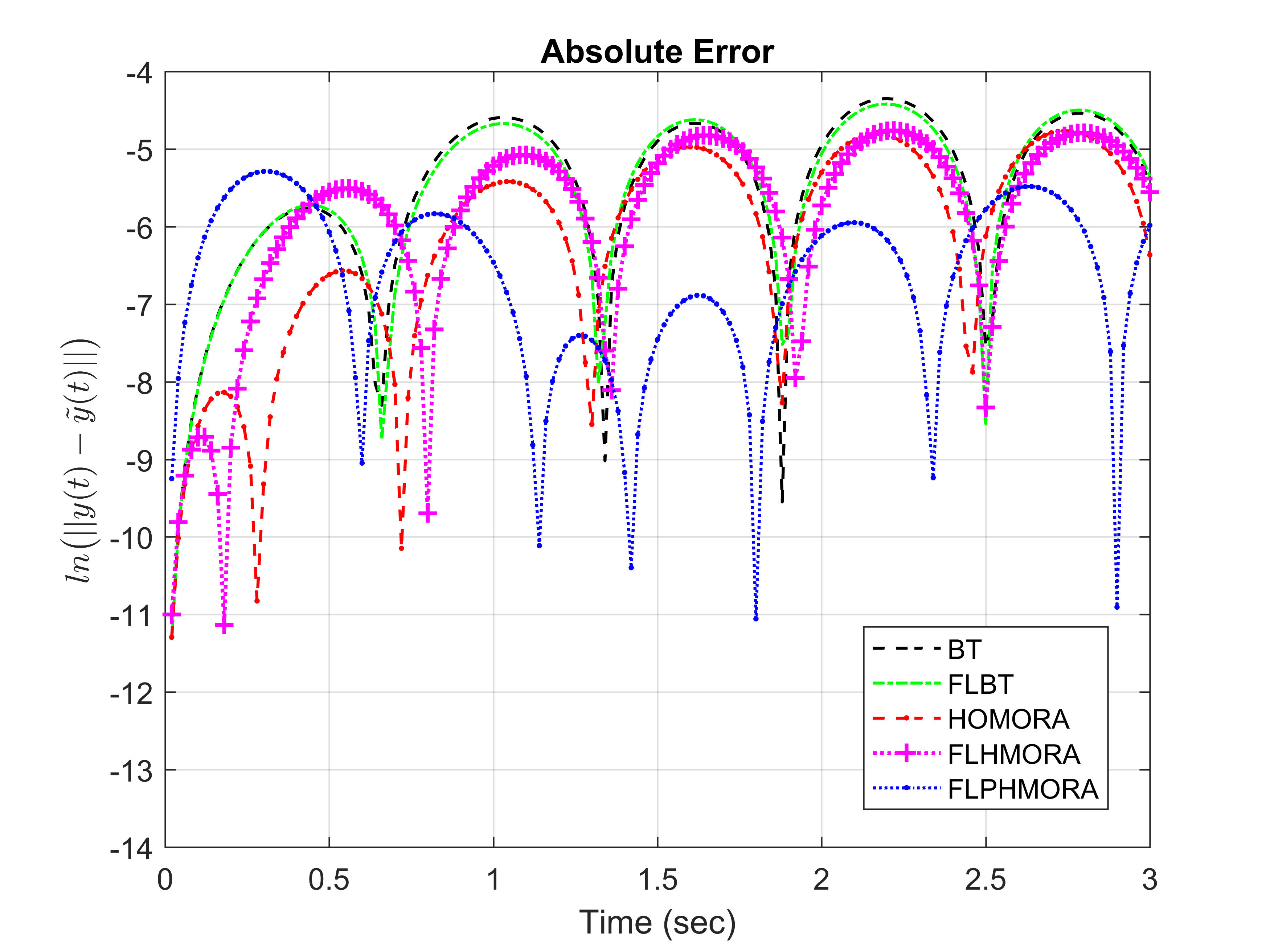}
\caption{$ln\big(||y(t)-\tilde{y}(t)||\big)$ for the input $u(t)=0.01sin(5t)$}\label{fig1}
\end{figure} The ROM generated by the FLPHMORA satisfies the optimality conditions (\ref{1.47}) and (\ref{1.48}) exactly. For the ROM generated by the FLHMORA, $C\hat{P}_\omega-\tilde{C}\tilde{P}_\omega=0.1364$ and $\hat{Q}_\omega^TB+\tilde{Q}_\omega\tilde{B}=0.0457$. Table \ref{tab0} compares the approximation error $||\Sigma-\tilde{\Sigma}||_{\mathcal{H}_{2,\omega}}$, and it can be seen that the FLPHMORA has minimum error as compared to other methods. The frequency-domain responses of the linearized error transfer functions are also plotted in Figure \ref{fig2a}, and here also, the FLPHMORA ensures the least error.
\begin{figure}[!h]
\centering
\includegraphics[width=12cm]{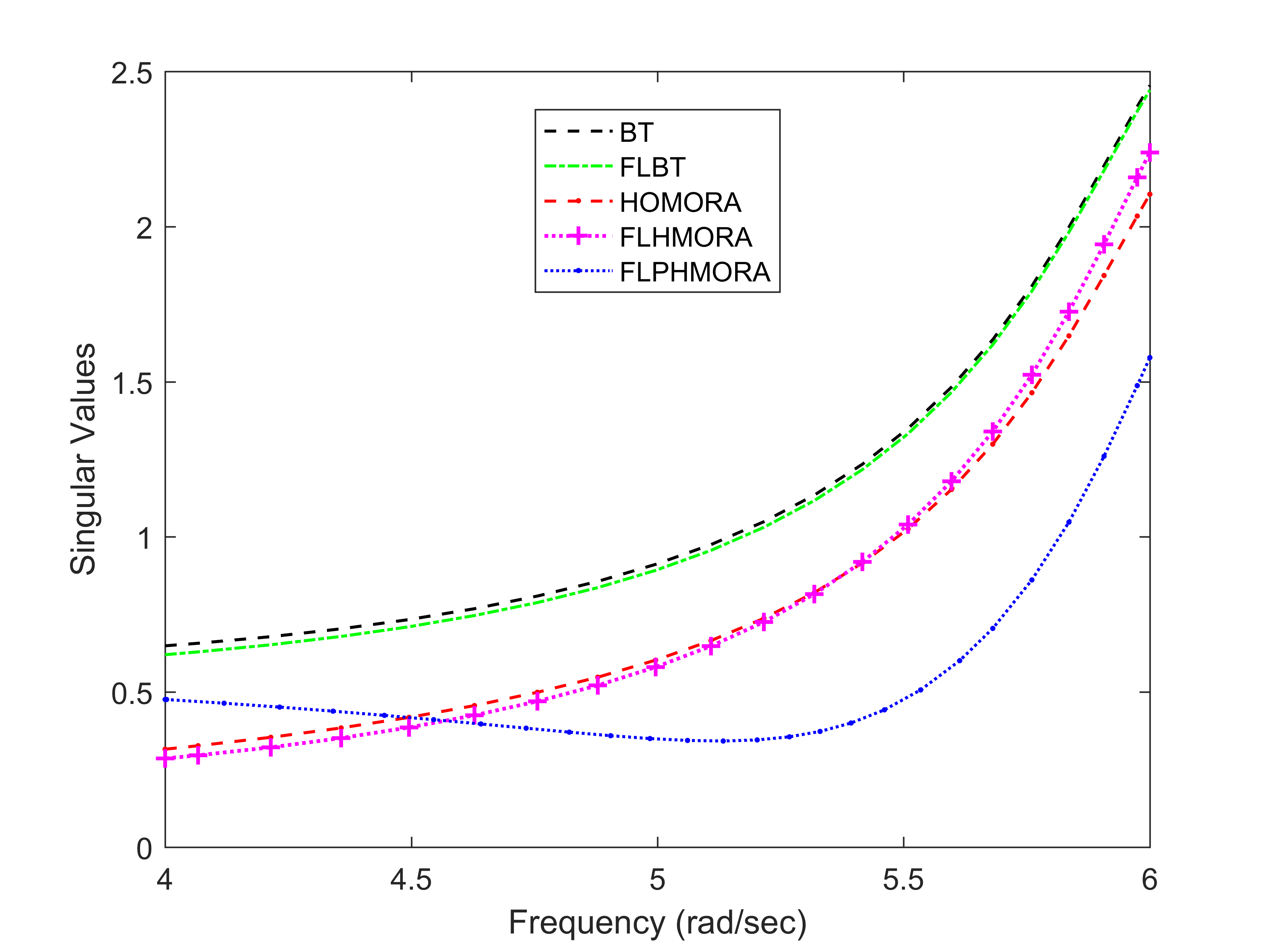}
\caption{Singular values of the linearized $\Sigma_e$ within $[4,6]$ rad/sec}\label{fig2a}
\end{figure}
\begin{table}[h]
\centering
\caption{Error Comparison: $||\Sigma-\tilde{\Sigma}||_{\mathcal{H}_{2,\omega}}$}\label{tab0}
\begin{tabular}{|c|c|c|c|c|}
\hline
\textbf{BT}     & \textbf{FLBT}      & \textbf{HOMORA} & \textbf{FLHMORA}   & \textbf{FLPHMORA} \\ \hline
\multicolumn{5}{|c|}{\textbf{Illustrative Example}}          \\ \hline
1.1995 & 1.1893    & 1.0302 & 1.0318    & 0.8640    \\ \hline
\multicolumn{5}{|c|}{\textbf{Power System Example}}          \\ \hline
0.0196 & 0.0026    & 0.0186 & 0.0026    & 0.0021    \\ \hline
\multicolumn{5}{|c|}{\textbf{Heat Transfer Example}}         \\ \hline
0.0015 & 1.5351$\times 10^{-5}$ & 0.0029 & 1.5344$\times 10^{-5}$ & $51.5344\times 10^{-5}$ \\ \hline
\end{tabular}
\end{table}

Next, we obtain $3^{rd}$ order ROMs using the BT, the TLBT, the HOMORA, the TLHMORA, and the TLPHMORA. We set the desired time interval as $[0,0.5]$ sec in the TLBT, the TLHMORA, and the TLPHMORA to ensure good accuracy within $[0,0.5]$ sec. The absolute error in the output response is compared in Figure \ref{fig2} on a logarithmic scale, and it can be seen that the time-limited MOR algorithms provide the best approximation within the desired time interval.
\begin{figure}[!h]
\centering
\includegraphics[width=12cm]{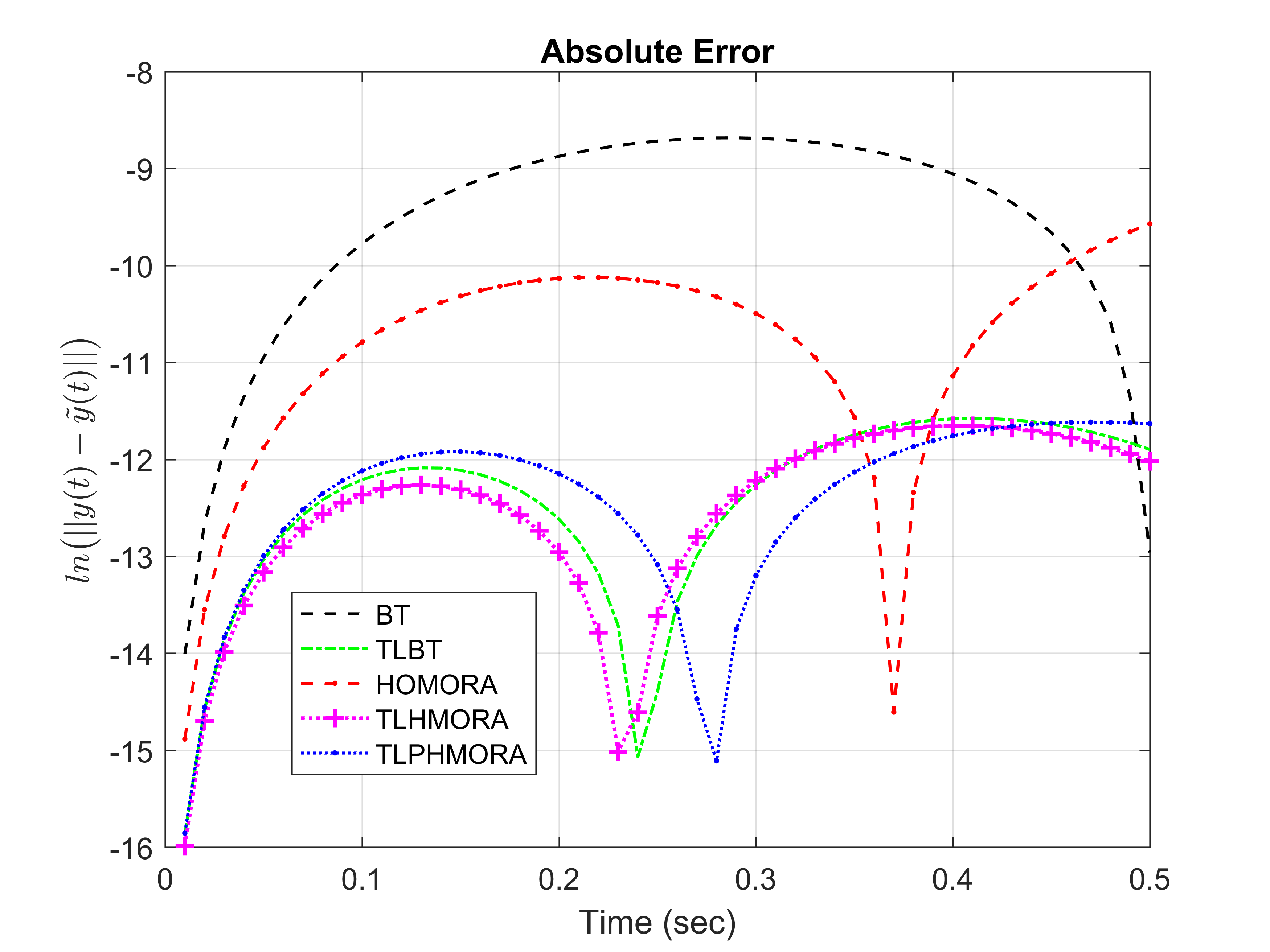}
\caption{$ln\big(||y(t)-\tilde{y}(t)||\big)$ for the input $u(t)=0.01sin(5t)$ within $[0,0.5]$ sec}\label{fig2}
\end{figure} The ROM generated by the TLPHMORA satisfies the optimality conditions (\ref{2.18}) and (\ref{2.19}) exactly. For the ROM generated by the TLHMORA, $C\hat{P}_\tau-\tilde{C}\tilde{P}_\tau=\begin{bmatrix} 0.0008  &  0.0022 &  -0.0005\end{bmatrix}$ and $\hat{Q}_\tau^TB+\tilde{Q}_\tau\tilde{B}=\begin{bmatrix}-0.0006 & -0.0014 &0.0002\end{bmatrix}^T$. Table \ref{tab1} compares the approximation error $||\Sigma-\tilde{\Sigma}||_{\mathcal{H}_{2,\tau}}$, and it can be noted that the TLHMORA and the TLPHMORA offer the least error.
\begin{table}[h]
\centering
\caption{Error Comparison: $||\Sigma-\tilde{\Sigma}||_{\mathcal{H}_{2,\tau}}$}\label{tab1}
\begin{tabular}{|c|c|c|c|c|}
\hline
\textbf{BT}     & \textbf{TLBT}      & \textbf{HOMORA} & \textbf{TLHMORA}   & \textbf{TLPHMORA} \\ \hline
\multicolumn{5}{|c|}{\textbf{Illustrative Example}}          \\ \hline
0.0850 & 0.0135    & 0.0385 & 0.0125    & 0.0121    \\ \hline
\multicolumn{5}{|c|}{\textbf{Power System Example}}          \\ \hline
0.0349 & 0.0068    & 0.0331 & 0.0065    & 0.0059    \\ \hline
\end{tabular}
\end{table}

\textbf{Power System Example:} Consider a $17^{th}$ order power system model from \citep{al1993new} that has $4$ inputs and $3$ outputs. This model has been used as a benchmark in the literature; see for instance \citep{breiten2012interpolation,ahmad2017implicit}. Let the input signal be a sinusoid with a frequency and amplitude of $3$ rad/sec and $0.01$, respectively, i.e., $u(t)=0.01sin(3t)$. We obtain $9^{th}$ order ROMs using the BT, the FLBT, the HOMORA, the FLHMORA, and the FLPHMORA. We set the desired frequency interval as $[2,4]$ rad/sec in the FLBT, the FLHMORA, and the FLPHMORA to ensure good accuracy at and in close neighbourhood of $3$ rad/sec. The absolute error in the output response (corresponding to the first output) is compared in Figure \ref{fig3} on a logarithmic scale, and it can be seen that the frequency-limited MOR algorithms provide the best approximation. The frequency-domain responses of the linearized error transfer functions are plotted in Figure \ref{fig3a}. It is evident from Figure \ref{fig3a} that the FLPHMORA ensures the least error.
\begin{figure}[!h]
\centering
\includegraphics[width=12cm]{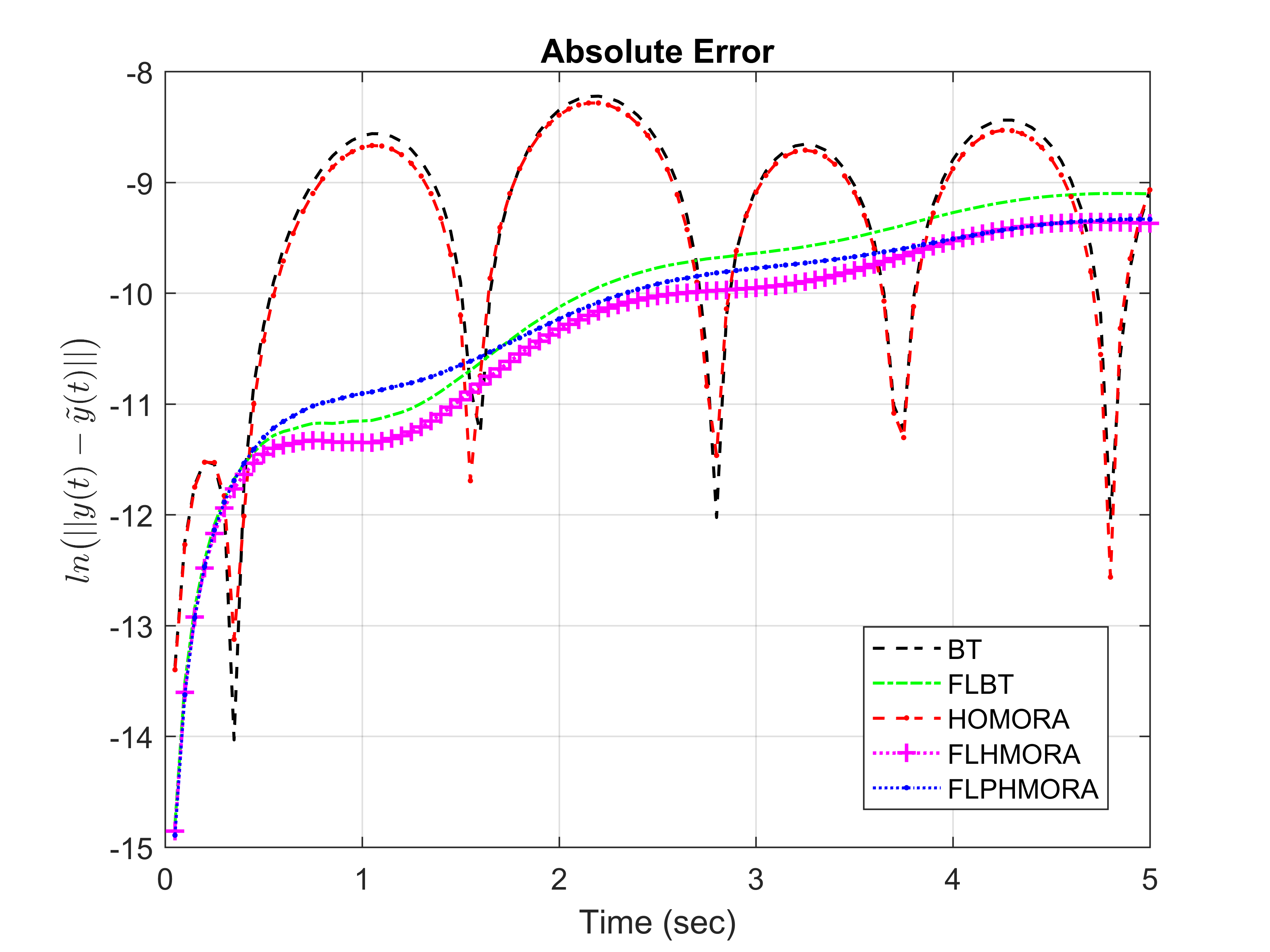}
\caption{$ln\big(||y(t)-\tilde{y}(t)||\big)$ for the input $u(t)=0.01sin(3t)$}\label{fig3}
\end{figure}
\begin{figure}[!h]
\centering
\includegraphics[width=12cm]{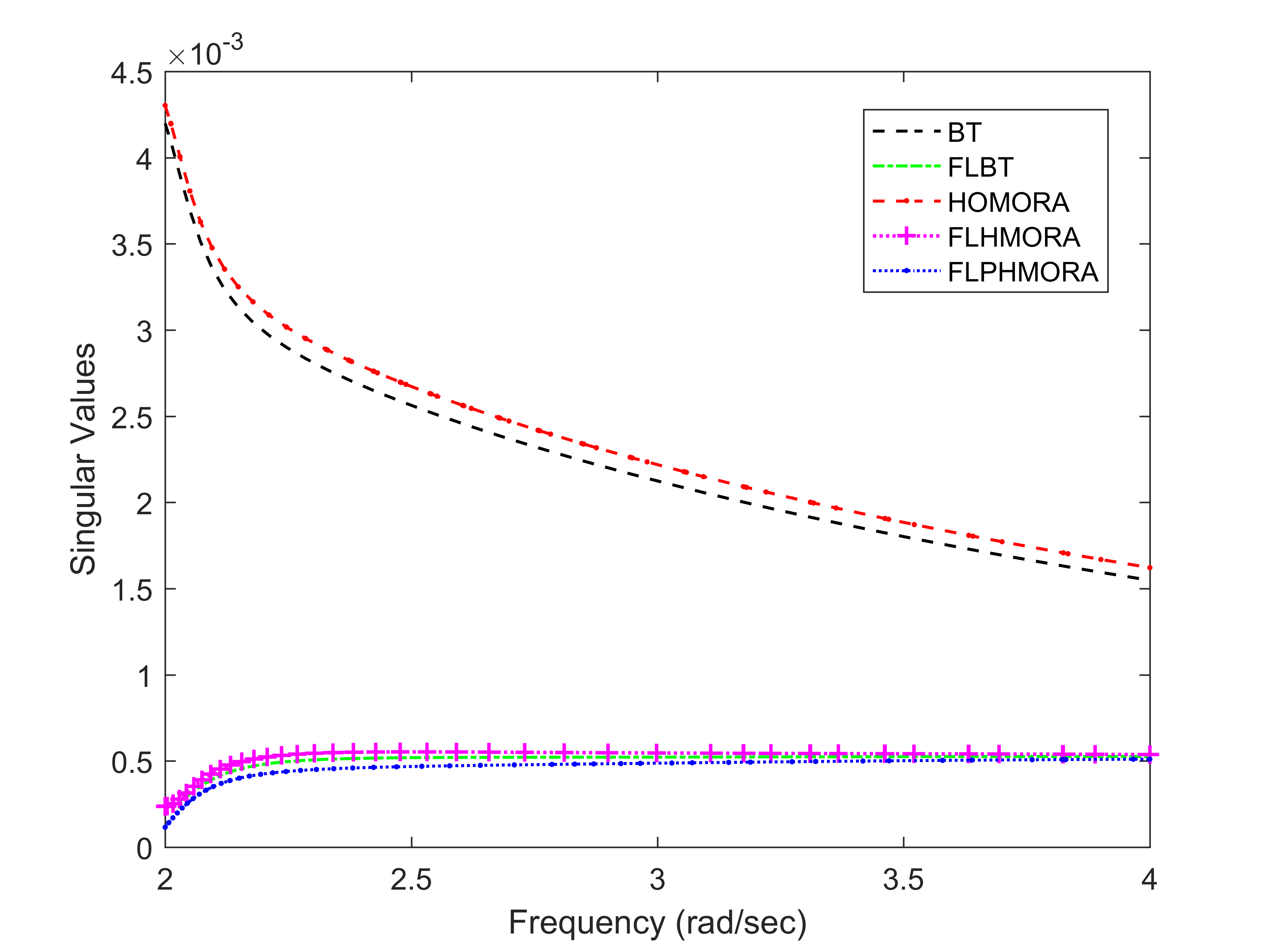}
\caption{Singular values of the linearized $\Sigma_e$ within $[2,4]$ rad/sec}\label{fig3a}
\end{figure} Table \ref{tab0} compares the approximation error $||\Sigma-\tilde{\Sigma}||_{\mathcal{H}_{2,\omega}}$, and it can be seen that the FLPHMORA yields the least error.

Next, we obtain $9^{th}$ order ROMs using the TLBT, the TLHMORA, and the TLPHMORA. We set the desired time interval as $[0,2]$ sec in the TLBT, the TLHMORA, and the TLPHMORA to ensure good accuracy within $[0,2]$ sec. The absolute error in the output response (corresponding to the first output) is compared in Figure \ref{fig4} on a logarithmic scale, and it can be seen that the time-limited MOR algorithms provide the best approximation within the desired time interval.
\begin{figure}[!h]
\centering
\includegraphics[width=12cm]{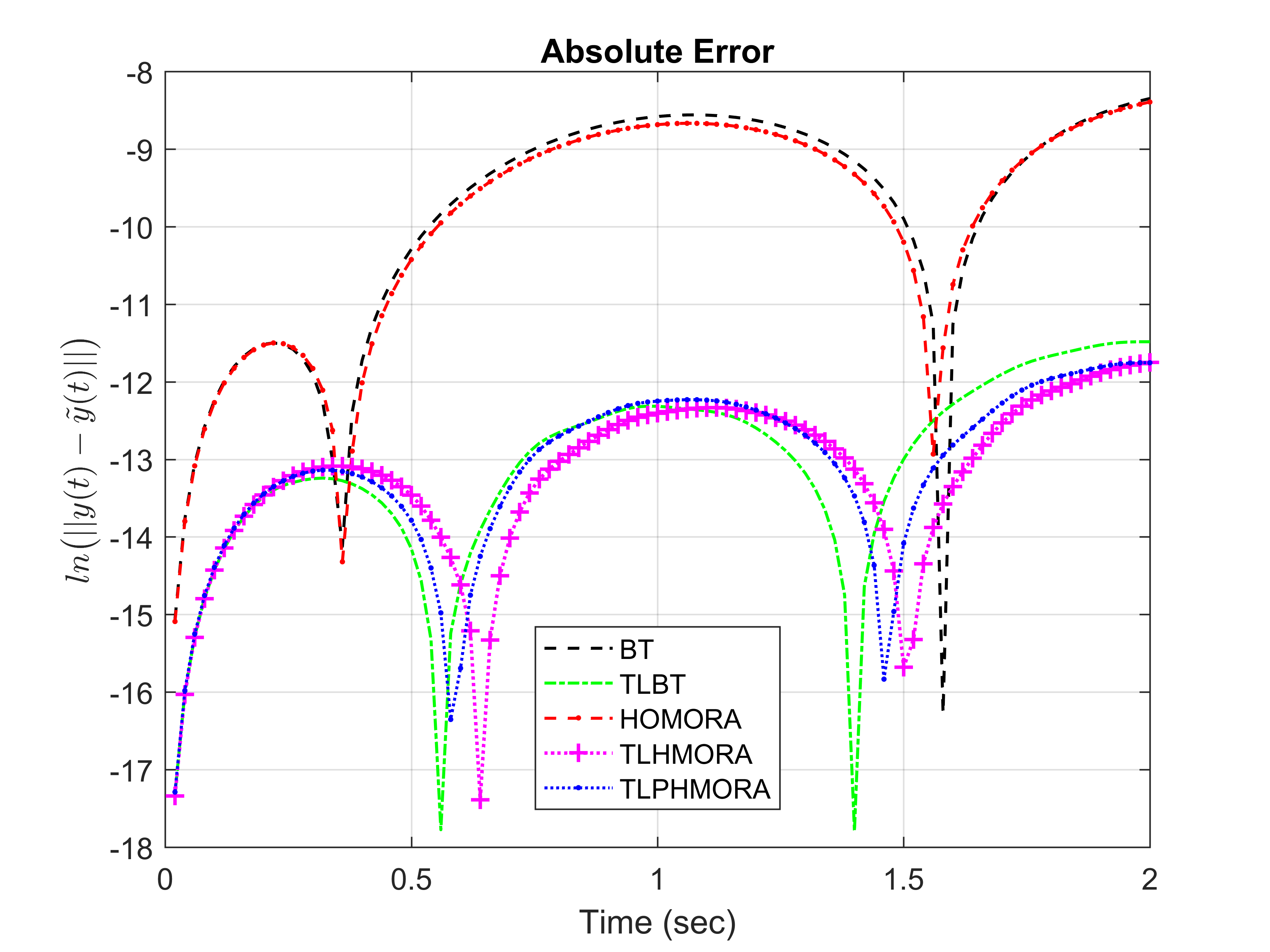}
\caption{$ln\big(||y(t)-\tilde{y}(t)||\big)$ for the input $u(t)=0.01sin(3t)$ within $[0,2]$ sec}\label{fig4}
\end{figure} Table \ref{tab1} compares the approximation error $||\Sigma-\tilde{\Sigma}||_{\mathcal{H}_{2,\tau}}$, and it can be noted that the TLHMORA and the TLPHMORA offer the least error.

\textbf{Heat Transfer Example:} Consider the boundary controlled heat transfer model, which is considered as a benchmark problem in the literature \citep{benner2011lyapunov,breiten2012interpolation,ahmad2017implicit,xu2017approach}. The spatial discretization of the model used in \citep{ahmad2017implicit} using $2500$ grid points yields a $2500^{th}$ order single-input single-output bilinear system. Let the input signal be a sinusoid with a frequency and amplitude of $1$ rad/sec and $0.01$, respectively, i.e., $u(t)=0.01sin(1t)$. We obtain $1^{st}$ order ROMs using the BT, the FLBT, the HOMORA, the FLHMORA, and the FLPHMORA. We set the desired frequency interval as $[0,2]$ rad/sec in the FLBT, the FLHMORA, and the FLPHMORA to ensure good accuracy at and in close neighbourhood of $1$ rad/sec. The absolute error in the output response is compared in Figure \ref{fig5}, and it can be seen that the frequency-limited MOR algorithms provide the best approximation. The frequency-domain responses of the linearized error transfer functions are plotted in Figure \ref{fig4a}. It is evident from Figure \ref{fig4a} that the FLPHMORA ensures good accuracy.
\begin{figure}[!h]
\centering
\includegraphics[width=12cm]{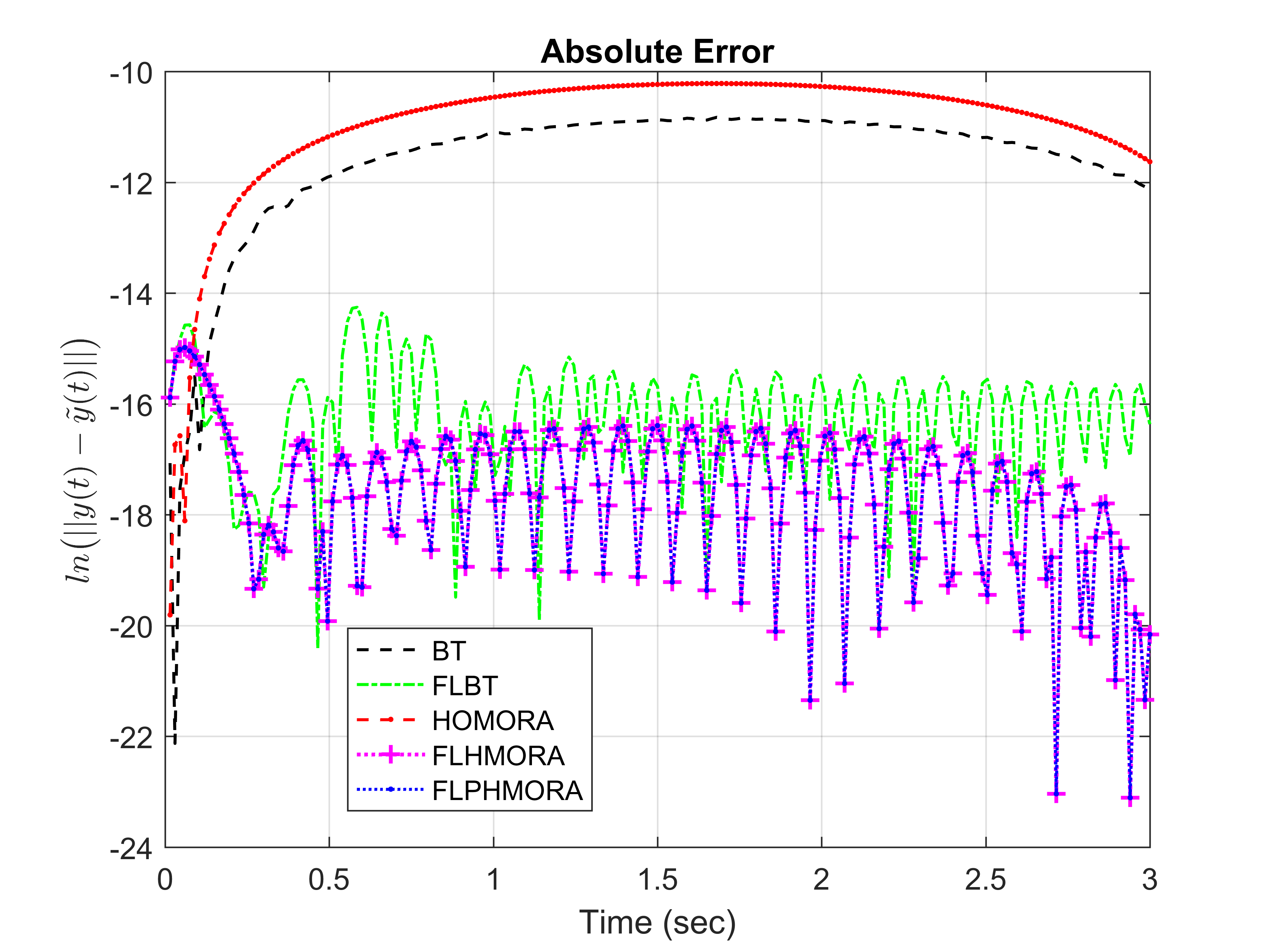}
\caption{$ln\big(||y(t)-\tilde{y}(t)||\big)$ for the input $u(t)=0.01sin(1t)$}\label{fig5}
\end{figure}
\begin{figure}[!h]
\centering
\includegraphics[width=12cm]{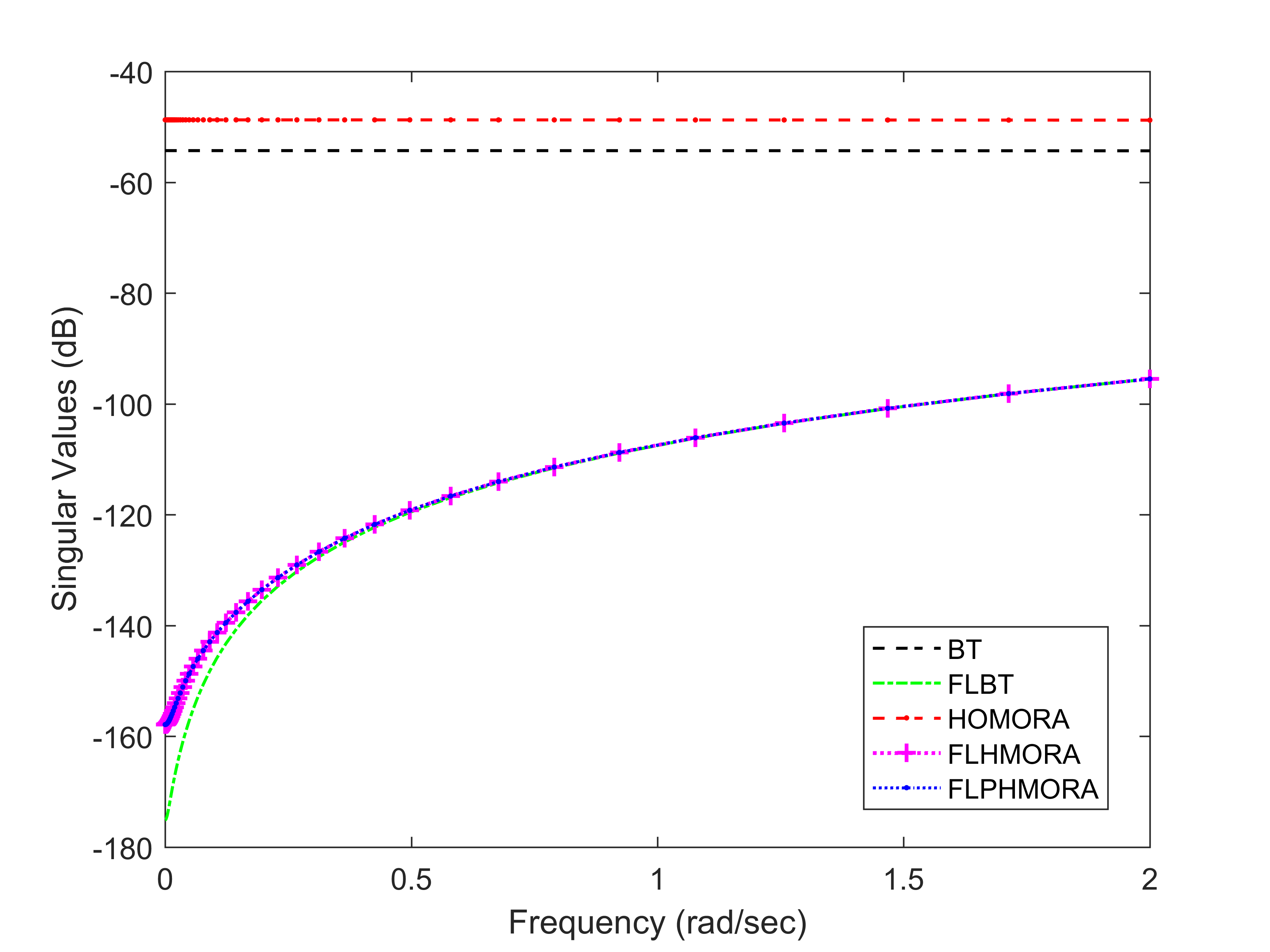}
\caption{Singular values of the linearized $\Sigma_e$ within $[0,2]$ rad/sec}\label{fig4a}
\end{figure} Table \ref{tab0} compares the approximation error $||\Sigma-\tilde{\Sigma}||_{\mathcal{H}_{2,\omega}}$, and it can be seen that the FLPHMORA yields the least error. Table \ref{taba} compares the simulation time, and it can be seen that the FLPHMORA takes the least time for execution.
\begin{table}[h]
\centering
\caption{Simulation Time (sec)}\label{taba}
\begin{tabular}{|c|c|c|c|c|}
\hline
\textbf{BT}     & \textbf{FLBT}      & \textbf{HOMORA} & \textbf{FLHMORA}   & \textbf{FLPHMORA} \\ \hline
$543.95$  & $706.04$     & $292.59$  & $278.37$     & $91.84$    \\ \hline
\end{tabular}
\end{table}

Next, we set the grid points to $529$, which results in $529^{th}$ order single-input single-output bilinear system. We obtain $1^{st}$ order ROMs using the TLBT, the TLHMORA, and the TLPHMORA. We set the desired time interval as $[0.5,1.5]$ sec in the TLBT, the TLHMORA, and the TLPHMORA to ensure good accuracy within $[0.5,1.5]$ sec. The absolute error in the output response is compared in Figure \ref{fig6} on a logarithmic scale, and it can be seen that the time-limited MOR algorithms provide the best approximation within the desired time interval.
\begin{figure}[!h]
\centering
\includegraphics[width=12cm]{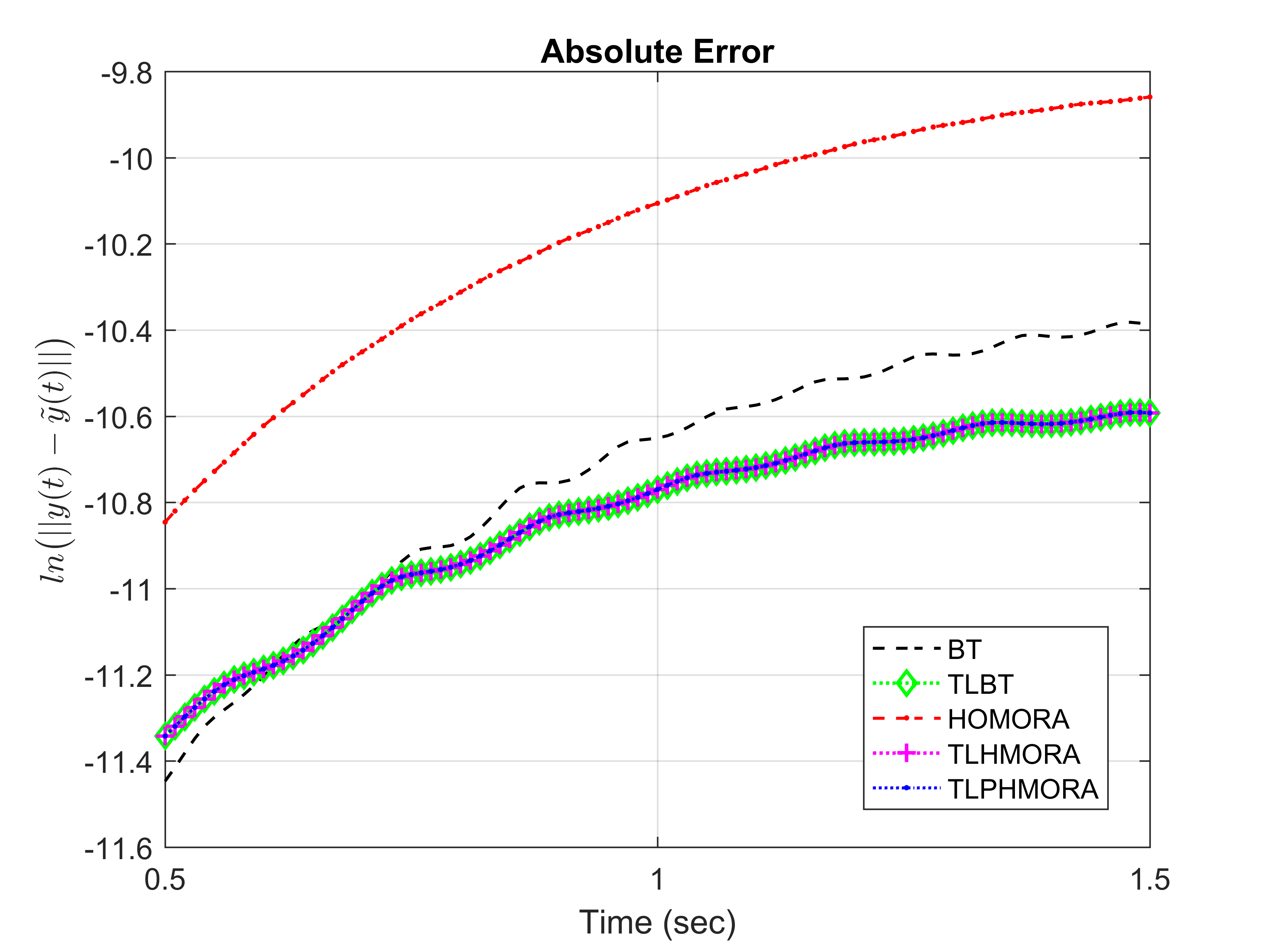}
\caption{$ln\big(||y(t)-\tilde{y}(t)||\big)$ for the input $u(t)=0.01sin(1t)$ with $[0.5,1.5]$ sec}\label{fig6}
\end{figure} The value of $||\Sigma||_{\mathcal{H}_{2,\tau}}$ in this example is quite small. Therefore, we compare the relative errors in this example for clarity. Table \ref{tab6} compares the relative error $\frac{||\Sigma-\tilde{\Sigma}||_{\mathcal{H}_{2,\tau}}}{||\Sigma||_{\mathcal{H}_{2,\tau}}}$, and it can be noted that the TLHMORA and the TLPHMORA offer the least error.
\begin{table}
\caption{\textbf{Heat Transfer Example:} Error Comparison: $\frac{||\Sigma-\tilde{\Sigma}||_{\mathcal{H}_{2,\tau}}}{||\Sigma||_{\mathcal{H}_{2,\tau}}}$}\label{tab6}
\centering
\begin{tabular}{|c|c|c|c|c|}
\hline
\textbf{BT}     & \textbf{TLBT}      & \textbf{HOMORA} & \textbf{TLHMORA}   & \textbf{TLPHMORA} \\ \hline
0.9965 & $1.6369\times 10^{-7}$ & 0.9999 & $1.6369\times 10^{-7}$ & $1.6369\times 10^{-7}$\\\hline
\end{tabular}
\end{table} Table \ref{tabb} compares the simulation time, and it can be seen that the TLPHMORA takes the least time for execution.
\begin{table}[h]
\centering
\caption{Simulation Time (sec)}\label{tabb}
\begin{tabular}{|c|c|c|c|c|}
\hline
\textbf{BT}     & \textbf{TLBT}      & \textbf{HOMORA} & \textbf{TLHMORA}   & \textbf{TLPHMORA} \\ \hline
$3.42$ & $3.96$     & $2.31$  & $0.45$     & $0.42$    \\ \hline
\end{tabular}
\end{table}
\section{Conclusion}
We formulate the time-limited $\mathcal{H}_2$-optimal MOR problem and derive first-order optimality conditions for the local optimum of the problem. We proposed a heuristic algorithm which attempts to generate a local optimum for the problem. We also proposed two new algorithms that generate a ROM, which satisfies a subset of the optimality conditions of the local optimum for the time-limited and frequency-limited $\mathcal{H}_2$-optimal MOR problems. Our algorithms are computational efficient and accurate as compared to the existing algorithms. The numerical simulation confirms the theoretical results proposed in the paper.
\section*{Funding}
This work is supported by the National Natural Science Foundation of China under Grant (No. $61873336$, $61873335$), and supported in part by $111$ Project (No. D$18003$). M. I. Ahmad is supported by the Higher Education Commission of Pakistan under the National Research Program for Universities Project ID $10176$.
\section*{Disclosure Statement}
The authors declare no conflict of interest.

\end{document}